\documentclass[11pt]{amsart}

\usepackage[utf8]{inputenc}
\usepackage[a4paper, margin=2cm]{geometry}
\usepackage{amsthm,amssymb,mathtools,amsfonts}
\usepackage[dvipsnames]{xcolor}
\usepackage[foot]{amsaddr}
\usepackage{tikz}
\usepackage{bbm} 
\usepackage{comment}
\usepackage{amsmath}
\usepackage{amsfonts}
\usepackage{amssymb}
\usepackage{amsthm}
\usepackage{pgfplots}
\usepackage{graphicx}
\usepackage{lmodern}
\usepackage{cancel}
\usepackage{physics}
\usepackage{mathtools}
\usepackage{dsfont}
\usepackage{comment}
\usepackage{algorithm2e}
\usepackage{tikz}
\usetikzlibrary{calc}
\usepackage{enumitem}
\usetikzlibrary{positioning}
\usepackage{capt-of}
\usepackage{subcaption}
\pgfplotsset{compat=1.18}

\usepackage{newfloat}
\DeclareFloatingEnvironment[
  fileext=lop,
  listname={List of Protocols},
  name=Protocol,
  placement=tbp
]{protocol}

\usepackage[colorlinks]{hyperref}
\hypersetup{
    colorlinks  = true,
    citecolor   = CadetBlue,
    linkcolor   = MidnightBlue,
    urlcolor    = PineGreen
}
\usepackage{thmtools}
\declaretheorem[parent=section]{theorem}
\declaretheorem[numberlike=theorem, name=Proposition]{proposition}

\declaretheorem[numberlike=theorem, name=Example]{example}
\declaretheorem[numberlike=theorem, name=Lemma]{lemma}
\declaretheorem[numberlike=theorem, name=Definition]{definition}
\declaretheorem[numberlike=theorem, name=Corollary]{corollary}

\declaretheorem[numberlike=theorem, name=Assumptions]{assumptions}

\DeclareMathOperator{\id}{id}

\title{Device independent quantum key distribution with robust self-tests}

\author{Andreas Bluhm$^1$}
\address{$^1$ Univ. Grenoble Alpes, CNRS, Grenoble INP, LIG, Grenoble, France}

\author{Gereon Ko\ss mann$^{2}$}
\address{$^2$ Institute for Quantum Information, RWTH Aachen University, Aachen, Germany}
\author{René Schwonnek$^{3}$}
\address{$^3$ Institute for theoretical Physics, Leibniz University Hannover, Germany}

\begin{document}

\begin{abstract}    
    Device-independent quantum key distribution (DIQKD) provides a model of quantum key distribution with minimal assumptions and highly abstract theoretical building blocks. Although DIQKD frees us from detailed discussions of specific device models and associated error parameters, it replaces them with fundamental assumptions about the validity of quantum experiments. In this work, we propose a way to lift a protocol based on DIQKD-style assumptions to a device-dependent QKD protocol by performing local self-tests in the laboratories of the two key-generating parties. In particular, we consider routed Bell-test setups as a means of self-testing the local parties in earnest and develop a rigorous mathematical framework showing that the underlying optimization problems can indeed be transferred to the device-dependent QKD setting. As an application, we illustrate many of the relevant techniques through the case study of a routed BB84 protocol.
\end{abstract}

\maketitle
\tableofcontents
\section{Introduction}

Quantum key distribution offers an opportunity to extract a provably secure key between distant parties, which is certified by quantum theory \cite{RENNER2008}. A standard model of QKD considers two parties, Alice and Bob, both equipped with quantum devices characterized by quantum measurements and sharing bipartite quantum states \cite{Ekert1991}. Usual QKD protocols now expect the two parties to measure the shared states often enough to obtain meaningful statistics, which can subsequently be used to extract a secure key by the rules of the composable security framework (cf.~e.g.~\cite{Tomamichel2017}). In particular, since its very first developments \cite{Bennett1992,Bennett_2014,Ekert1991}, it has been shown that actual security requires the consideration of, in principle, all imperfections of the underlying devices and mismatches between theory and experiment. In order to handle the huge number of potential security lacks, so-called device-independent quantum key distribution (DIQKD) emerged as the gold standard for QKD experiments, as it does not require any assumption on the underlying devices beyond that they are subject to the laws of quantum mechanics \cite{Mayers1998}. Moreover, in concrete bipartite scenarios, security proofs for DIQKD only require statistics made by the two partners, Alice and Bob, in order to estimate the amount of key that can be extracted \cite{pironio2009device,Vazirani_2014}. Given these arguably weak assumptions on the physical modeling, it has been discovered that, instead of a concrete mathematical modeling of the devices, a different aspect of the experiment needs to be well understood. At the forefront of this aspect are the loopholes of Bell experiments, as they turn out to be the most important security risks in a DIQKD experiment. Generally, a loophole can be defined as any lack of coherence between experiment and physical modeling that is able to fake the observed statistics. The most prominent instance is the detection loophole, as its effect on the observed statistics equals the influence of long-distance experiments, which, by the natural Lambert-Beer decay effect, have vanishing event sizes.
   
In a recent sequence of works \cite{Lim2013,Chaturvedi2024,Lobo2024,Le_Roy_Deloison_2025}, models of DIQKD have been developed in order to address the detection loophole, which have been named \emph{routed Bell tests}. Routed Bell tests in \cite{Lim2013,Chaturvedi2024,Lobo2024,Le_Roy_Deloison_2025} are tripartite experiments between Alice, Bob, and a third party named Fred in this work. It is assumed that Alice and Fred are close to each other such that one may assume that Alice and Fred are able to perform a close-to-perfect CHSH test, while Alice and Bob could be far apart and perform a key-generation protocol. The main purpose of routed Bell tests has been to show that the zero-key threshold, i.e., the amount of randomness needed in order to extract key, significantly improves depending on the detection efficiency. In comparison, \cite{kossmann2025routedbelltestsarbitrarily} introduced a fourth party, called George, and investigated the key rate depending on the quantum bit-error rate and local Bell violation, where both Alice and Bob perform local tests with Fred and George, respectively. The key insight compared to its predecessors is a marginal constraint between the applied states for key generation between Alice and Bob and the Bell tests between Alice and Fred, respectively between Bob and George, as we will discuss in the next section.

Even though it has been observed in all proposals \cite{Lim2013,Chaturvedi2024,Lobo2024,Le_Roy_Deloison_2025,kossmann2025routedbelltestsarbitrarily} that Fred and George are, at least in principle, local self-tests of Alice’s and Bob’s devices, a rigorous mathematical formulation is still missing. Self-testing is a technique in the realm of device-independent quantum information processing that allows one to infer directly from observed statistics a concrete mathematical model of the implemented states and measurements, up to isometric degrees of freedom \cite{Mayers2004,upi2020,Paddock_2023}. Thus, as observed in \cite{kossmann2025routedbelltestsarbitrarily}, in the case of a perfect CHSH violation, key rates as in a device-dependent situation are observed numerically \cite{Shor_2000}. For the case of a perfect self-test, it is well known that the optimal device-dependent key rates can be recovered \cite{Woodhead_2016}. However, in proper device-independent models, the performance of perfect self-tests has only very limited applications, as one can arguably assume that a perfect test cannot be recovered in practice, at the very least because of finite-size effects in the observed statistics. For this purpose, a notion of robust self-testing has recently been introduced, eventually capturing the case of imperfect statistics \cite{upi2020,operator_algebra_self_test2024}. We furthermore remark that the model \cite{kossmann2025routedbelltestsarbitrarily} was written in the language of $C^\star$-algebras, which fits the setup for \emph{abstract self-tests} in \cite{operator_algebra_self_test2024}. However, for our calculations here, it will be beneficial always to compute the relevant quantities in a representation without resorting to abstract $C^\star$-algebras. \cite[Thm.~3.8]{operator_algebra_self_test2024} shows that this perspective is equivalent to using the abstract self-test framework, at least for states in the minimal tensor product. 

This work is divided into two main parts. In the first part, \autoref{sec:models}, we compare the models of \cite{Lim2013,Lobo2024} with \cite{kossmann2025reliableentropyestimationobserved} and discuss the role of the switch that routes the states either to Alice and Fred or to Alice and Bob. In \autoref{example:assumptions_switch}, we show that a no-signaling assumption between Alice's device and the switch is crucial for the security of the setup. In particular, for a fully device-independent modeling, it is important to make the assumptions of the setup explicit in order to connect to concrete experiments. In particular, all setups -- whether those in \cite{Lim2013,Lobo2024} or, explicitly, \cite{kossmann2025routedbelltestsarbitrarily} -- need to fulfill the marginal constraint, that is, the incoming state in Alice's lab must always the same independently of the random variable corresponding to the switch. In the second part of this work, \autoref{sec:robust_lifts}, we use the language of robust self-tests from \cite{operator_algebra_self_test2024} and show that the setup with local self-tests in Alice's and Bob's labs can be lifted to an effectively device-dependent optimization problem even in the presence of errors, enabling us to use all the tools from device-dependent QKD.

\section{Models}\label{sec:models}
 \def\rs#1{{\color{PineGreen}RS:#1}}%
\begin{figure}[ht]
  \includegraphics[width=0.9\textwidth]{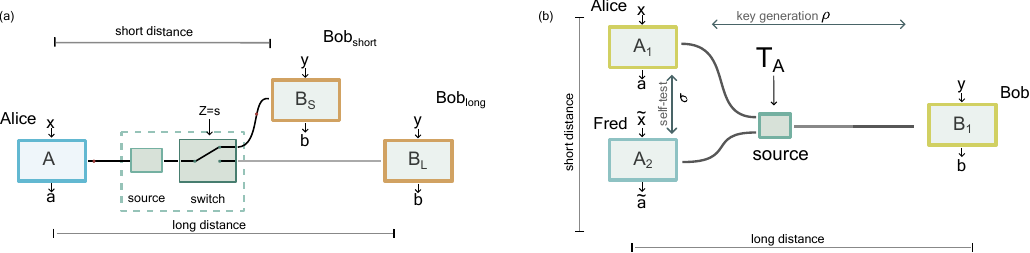}
  \caption{\label{fig:settings}  (a) Figure adapted from \cite{Lobo2024}.  In this model the switch and the source are treated as different devices that do not influence each other. The source creates a fixed state.  Controlled by a random variable $Z$, the Bob part of this state  is either routed to $Bob_{short}$ or to  $Bob_{long}$ via the switch.
  (b) Figure adapted from \cite{kossmann2025routedbelltestsarbitrarily}. Here the action of the switch is modeled as part of the source, which now has an extra classical parameter $T_A$ as input.
  }
\end{figure}

In this section, we compare two models for the \emph{routed Bell test} scenario proposed in \cite{Lim2013}. Here, a local Bell test ensures the security of a quantum key distribution protocol between two faraway parties, which leads to a protocol that can tolerate high losses in the quantum channel.  Thus, this approach enables protocols for \emph{device independent quantum key distribution} (DIQKD) that work over longer distances. However, there are different ways in which this approach has been be modeled.

The first setup, displayed in \autoref{fig:settings}~(a), has been introduced in \cite{Chaturvedi2024} and has been substantially improved in \cite{Lobo2024}. Its application to DIQKD has been studied in \cite{Le_Roy_Deloison_2025} and subsequent work. Based on this, a similar model, but using a different mathematical framework, i.e. the commuting operator framework, has been used in \cite{kossmann2025routedbelltestsarbitrarily} in order to extend the former work to a multi-party scenario on both sides. This setup is illustrated in \autoref{fig:settings}~(b). This section is devoted to elaborating on the differences between the setups in \autoref{fig:settings}~(a) and \autoref{fig:settings}~(b). 

\subsection{Security assumptions and switches}
DIQKD can arguably be understood as a framework in which the devices and any explicit classical description of the implemented measurements are abstracted away, and all assumptions are reduced to the fundamental ones: the correctness of quantum theory, access to local private randomness, and the confidentiality of any information within the agents' local laboratories. If a more intricate, possibly multipartite, model of a cryptographic protocol aims to call itself \emph{device-independent}, security under these assumptions provides the standard said protocol must rise to meet. 

In the concrete setting of routed Bell tests, we are therefore asked to clearly articulate the assumptions under consideration. From the setup it is clear that all proposed experiments may, without doubt, assume the correctness of quantum theory and the availability of local private randomness for Alice and Bob. Moreover, the assumption of confidentiality of any information within the laboratories of the key-generating parties, Alice and Bob, can be adopted, since it does not differ from the one in any bipartite DIQKD protocol studied extensively over the last decade (cf., e.g., \cite{ferradini2025definingsecurityquantumkey} for a recent discussion of assumptions in (DI)QKD). However, the manner in which the statistics obtained by quantum measurements are used must be considered with great care, since in both proposals \autoref{fig:settings}~(a) and \autoref{fig:settings}~(b) the statistics between $A-B_S$ in \autoref{fig:settings}~(a) and $A-F$ in \autoref{fig:settings}~(b) are intended to be exploited via a local Bell test in order to drastically improve the parties' knowledge about the underlying experimental devices. In particular, both protocols claim to yield a device-independent setup that enables the agents to extract in a so-called routed BB84 protocol\footnote{i.e. local CHSH tests and BB84 statistics between Alice and Bob} as much key as was previously known to be achievable only in the device-dependent case, or even in the strict semi-device-independent setting (cf.\ \cite{Le_Roy_Deloison_2025} for the setup in \autoref{fig:settings}~(a) and \cite{kossmann2025routedbelltestsarbitrarily} for the setup in \autoref{fig:settings}~(b), both claiming to achieve the Shor--Preskill key rate \cite{Shor_2000} in the routed setting. For the semi-device-independent case, see \cite{Woodhead_2016}.).\footnote{It is surprising, but well known, that the semi-device-independent BB84 protocol admits the same tight key rate as the device-dependent protocol.} Since it is moreover well known -- and readily verifiable -- that the BB84 protocol itself is not secure in a device-independent setting, as it admits a local-hidden-variable model and a perfect-leakage model in dimension four (see, e.g., \cite{pironio2009device}), a concise argument for its security under very precise assumptions is required. It is also crucial that this argument can be readily extended to all further, more involved, protocols. The core question coming out of this discussion regarding both proposals in \autoref{fig:settings}~(a) and in \autoref{fig:settings}~(b) can be summarized as follows: 
\begin{center}
    \emph{How can the information of the local Bell test can be incorporated in the key-generation protocol in a precise mathematical way?}
\end{center} 

The key decision to make concerning the use of information obtained from the local test lies in how the switch is modeled. However, since a switch is a very concrete optical element, we aim to adopt the mathematical language commonly used in the DIQKD literature to abstract these details away, namely the formalism of Ludwig boxes as introduced by Ludwig in \cite{Ludwig1985}. In this language, we simply model the switch as a source which, conditioned on a classical variable $T$, emits two types of states, $\rho$ and $\sigma$. 

The point we would like to stress as  particularly important is that in the setup \autoref{fig:settings}~(a) Alice's laboratory always receives the same state emitted by the source, independently of the position of the switch. However, if one draws a box around the entire switching setup, as one would like to do in a device-independent setting as in  \autoref{fig:settings}~(b), this information becomes inaccessible at the level of the box description. Thus, the statement that the source sends the same state to Alice's lab, regardless of the classical switch variable $T$, may have to be promoted to an explicit assumption. In the following, we exhibit a specific local-hidden-variable model showing that this assumption is indeed essential.

\begin{example}[Assumptions on switch]\label{example:assumptions_switch}
To show what goes wrong if we do not enforce the following condition, consider the following example, where the source produces the following states:
\begin{equation}
    \rho_{ABF} = \begin{cases} \dyad{\Omega}_{A_0F} \otimes \dyad{1}_{A_1} \otimes \sigma_B & T=1 \\
    \tau_{A_0B} \otimes \dyad{0}_{A_1} \otimes \mu_{F} & T=0 \end{cases}\,.
\end{equation}
Here, $\ket{\Omega}=\frac{1}{\sqrt{2}}(\ket{00} + \ket{11})$ is a maximally entangled state, $\tau=\frac{1}{2}(\dyad{00}+\dyad{11})$ is a classically correlated state, and $\sigma$, $\mu$ are arbitrary states. We have decomposed $A$ as $A_0A_1$. On $A_1$, the source sends to Alice the information whether a test round or a key round is being played such that Alice's devices can change their behavior accordingly. It is easy to see that in the case $T=1$, a local CHSH inequality will be maximally violated and that in the case $T=0$, an eavesdropper holding a purification of $\tau$ will have full information about the key being generated. Importantly, we observe that the marginals on Alice's side $\rho_A$ depend on the random variable $T.$
\end{example}

Indeed, upon revisiting the argument, \cite{kossmann2025routedbelltestsarbitrarily} proposes a DIQKD model with the natural marginal constraint that ensures Alice's lab does not receive any information about whether a given round is a local-test round or a key-generation round. The same property follows from the concrete modeling of the switch in \cite[Fig.~1]{Lobo2024}. To be concrete, the comparable model from \cite{kossmann2025routedbelltestsarbitrarily} in the setting of one local party is drawn in \autoref{fig:settings}~(b).

Here the random variable $T$ is responsible for deciding whether a state $\rho_{ABF}$, commonly used for key-generation, or a state $\sigma_{ABF}$, used for the local test, is prepared. In order to make the protocol secure, it is enough to ensure that 
\begin{align}\label{eq:marginal_constraint}
    \rho_A = \sigma_A,
\end{align}
i.e., Alice can not distinguish between key-generation rounds and local self-tests. As shown in \cite{kossmann2025routedbelltestsarbitrarily}, this is a natural assumption in order to even come up with analytical key rate results (cf.~\autoref{eq:keyrate_beta_explicit_eps}) and \autoref{eq:marginal_constraint} becomes a natural constraint  when calculating key rates with non-commutative polynomial optimization techniques as in \cite{kossmann2025reliableentropyestimationobserved}. 

\subsection{Comparing statistics}

In this section, we will demonstrate that the set of correlations produced by the models in \autoref{fig:settings}~(a) and \autoref{fig:settings}~(b) with marginal constraint \autoref{eq:marginal_constraint} are identical. To this end, we will define two types of channels which formalize the state preparation in the two models. Here, $T$ is the classical register which decides which kind of state is prepared and $A'$ is the register prepared by the source which is not sent to Alice. The first type of channel we consider (corresponding to the model in \autoref{fig:settings}~(b)) is given as
\begin{equation}\label{eq:our_channel_model}
\begin{aligned}
    \Phi_{AA^\prime T\to ABF}: \mathcal{S}(\mathcal{H}_{A} \otimes \mathcal{H}_{A^\prime} \otimes \mathcal{H}_T) &\to \mathcal{S}(\mathcal{H}_{A} \otimes \mathcal{H}_{B} \otimes \mathcal{H}_F)  \\
    \rho_{AA^\prime T } &\mapsto \sum_t \Gamma_{AA^\prime \to ABF}^t\left(\bra{t} \rho_{AA^\prime R}\ket{t}\right) \,.
\end{aligned}
\end{equation}
Here, the $\Gamma_{AA^\prime \to ABF}^r: \mathcal S(\mathcal H_A \otimes \mathcal H_{A'}) \to \mathcal S(\mathcal H_A \otimes \mathcal H_{B} \otimes \mathcal H_F) $ form a collection of quantum channels indexed by a finite set of classical indices $t$ such that 
\begin{equation} \label{eq:marginal-constraint-channel}
\operatorname{tr}_{BF} \; \circ \; \Gamma_{AA^\prime \to ABF}^t = \operatorname{tr}_{BF} \; \circ \; \Gamma_{AA^\prime \to ABF}^{t^\prime}    
\end{equation}
for all $t,t^\prime \in T$. Then, \autoref{eq:marginal-constraint-channel} is a no-signaling constraint, enforcing in particular a marginal constraint such as in \autoref{eq:marginal_constraint}. 

The second type of channel is given as
\begin{equation}\label{eq:pironio_channel_model}
    \begin{aligned}
        \Phi_{A^\prime T\to BF}^\prime: \mathcal{S}(\mathcal{H}_{A^\prime} \otimes \mathcal{H}_T) &\to \mathcal{S}(\mathcal{H}_{B} \otimes \mathcal{H}_F). 
    \end{aligned}
\end{equation}
These channels correspond to the model in \autoref{fig:settings}~(a), where we have identified $B_{S}$ with $F$ and $B_{L}$ with $B$ to formulate both setups using the same systems.
\begin{proposition}\label{prop:comparison_models}
    In the device independent setting are both models \autoref{eq:our_channel_model} and \autoref{eq:pironio_channel_model} equivalent. That means, given a state $\rho_{ABF}$ as an output state of the first model, then there exists a channel of the second model and an input state $\tau_{AA^\prime T}$ of the second model such that $\rho_{ABF} = \operatorname{id}_A \otimes \Phi_{A^\prime T\to BF}^\prime(\tau_{AA^\prime T})$ and vice versa. 
\end{proposition}
\begin{proof}
    To see that the second model implies the first one we can just choose
    \begin{align}
        \Phi_{AA^\prime T\to ABF} = \operatorname{id}_A \otimes \Phi_{A^\prime T\to BF}^\prime.    
    \end{align}
    As $T$ is classical, we identify
    \begin{equation}
        \Gamma_{AA^\prime \to ABF}^t(\mu_{AA'}) = \operatorname{id}_A \otimes \Phi_{A^\prime T\to BF}^\prime(\mu_{AA'} \otimes \dyad{t}) 
    \end{equation}
    for any state $\mu_{AA'} \in \mathcal S(\mathcal H_A \otimes \mathcal H_{A'})$. The condition that the conditional channels $\operatorname{tr}_{BF} \; \circ \;\Gamma_{AA^\prime \to ABF}^t$ are independent of $t$ is trivially satisfied. Thus, the channels in the second model are special cases of the channels allowed in the  first model.

    Let us now consider the converse, i.e., that the first model implies the second one. As $\rho_{ABF}$ is an output state in the first model, there exists a state $\tau_{AA^\prime T} = \sum_t p_t \tau_{AA^\prime}^t\otimes \ketbra{t}{t}$ such that 
    \begin{align}
       \rho_{ABF} = \Phi_{AA^\prime T\to ABF}(\tau_{AA^\prime T}).
    \end{align}
    We consider in the following 
    \begin{align}
        \rho_{ABF} = \sum_t  p_t  \underbrace{\Gamma_r\left(\tau_{AA^\prime}^t\right)}_{\eqqcolon \rho_{ABF}^t}.
    \end{align}
    Here, we have written 
    \begin{equation*}
        \tau_{AA^\prime T} = \sum_{t} p_t\tau_{AA^\prime}^t \otimes \dyad{t},
    \end{equation*}
    which we can do since $T$ is a classical register.
    
    Now, because of $\operatorname{tr}_{BF} \; \circ \; \Gamma_{AA^\prime \to ABF}^t = \operatorname{tr}_{BF} \; \circ \; \Gamma_{AA^\prime \to ABF}^{t^\prime}$ for all $t,t^\prime \in T$, it holds that $\rho_A^t =  \rho_A$ for all $t\in T$. Hence, we can find a purifying system $E$ and purifications $\psi_{ABFE}^t$ of $\rho_{ABF}^t$ such that there exist by Uhlmann's theorem unitaries $U_{BFE}^t$ relating $\psi_{ABFE}$, a purification of $\rho_{ABF}$, to $\psi_{ABFE}^t$. Thus, we have an extended state 
    \begin{equation}
    \begin{aligned}
        \rho_{ABFE} &\coloneqq   \sum_t p_t \ketbra{\psi_{ABFE}^t}{\psi_{ABFE}^t} \\
        &=\sum_t p_t U_{BFE}^t\ketbra{\psi_{ABFE}}{\psi_{ABFE}}(U_{BFE}^t)^{\dagger} \\
        &= \operatorname{tr}_T\bigg[\bigg(\sum_{t}U_{BFE}^t \otimes \ketbra{t}{t}\bigg) \sum_t p_t \ketbra{\psi_{ABFE}}{\psi_{ABFE}} \otimes \ketbra{t}{t} \bigg(\sum_{t}U_{BFE}^t \otimes \ketbra{t}{t}\bigg)^\dagger\bigg].
    \end{aligned}
    \end{equation}
    Now we define as input state $\tau^\prime_{AA^\prime T}\coloneqq   \sum_t p_t \ketbra{\psi_{ABFE}}{\psi_{ABFE}} \otimes \dyad{t} $ with $A^\prime \coloneqq   BFE$. Then we apply the channel 
    \begin{equation}
    \operatorname{tr}_{ET}\circ \bigg(\sum_{t}U_{BFE}^t \otimes \ketbra{t}{t}\; \cdot \; \left(\sum_{t}U_{BFE}^t \otimes \ketbra{t}{t}\right)^\dagger\bigg)\, ,
    \end{equation} \label{eq:channel-for-converse}
    which outputs $\rho_{ABF}$ and the channel just acts on $A^\prime$ and the classical register $T$. Thus, we have found a suitable $\Phi_{A^\prime T\to BF}^\prime$ by using \autoref{eq:channel-for-converse} as its definition. 
\end{proof}
\begin{corollary}
    The construction from \autoref{prop:comparison_models} works even conditioned on an event $T = t$. 
\end{corollary}
Thus, we have clarified that both models we are considering yield the same statistics.

\subsection{Short-range quantum correlations in the multipartite setting}\label{subsec:srq} After debating models in earlier works, in \cite{Lobo2024} a proper definition for the set of short quantum correlations was given. Short range quantum correlations in the sense of \cite{Lobo2024} are the correlations in the model depicted in \autoref{fig:settings}~(a) that can be obtained without any entanglement being distributed to $Bob_{long}$. Based on the model considered in this work we are in a position to generalize this notion to situations in which Alice and Bob are assisted by arbitrarily many local partners, whereas the original definition in \cite{Lobo2024} did not allow for such a generalization. We will see that the previous definition will arise as a subcase. 

For the following consider a setting with parties $A_0\dots A_n$ on Alice's side and parties $B_0 \dots B_n$ on Bobs side. Routing is controlled by switches with input random variables $T_A$ and $T_B$ with support $\{0,1,\dots,n\}$. Employing the perspective of a virtual global source, this corresponds to the creation of a state $\omega^s$ for every input $(T_A,T_B)=s$.  
The measurement of each party on Alice's side has an input $x_k$ and an output $a_k$ for $k = 0,1\dots,n$. Inputs and outputs on Bob's side are labeled by $b_k$ and $y_k$, respectively. 
The measurements themselves are described by POVM elements $M^k_{a_k|x_k}$ and $N^k_{b_k|y_k}$. In a device-independent setting these can be assumed to be projective \cite{Schwonnek2021}. We will model the independence of the local measurements by assuming the typical commutator relations of a form
\begin{align}
    [M^k_{a_k|x_k},M^j_{a_j|x_j}]=[N^k_{b_k|y_k},N^j_{b_j|y_j}]=[M^k_{a_k|x_k},N^j_{b_j|y_j}]=0 . 
\end{align}
Formally, we can now treat this as a description of an universal $C^*$-algebra generated by partially commuting projections \cite{Blackadar2006}. Intuitively, this basically corresponds to a proper abstract description of what it means to consider all measurements of all states in all feasible Hilbert spaces. In this language a state can be seen as a linear functional acting on a measurement which motivates the notation $\omega^s(M)$ to describe the expectation of an observable $M$ in a state $\omega^s$. The notation  of  $\tr(\rho_\omega M_\mathcal{H})$ which is more common in the DIQKD community  corresponds to a  representation of this algebra on a concrete Hilbert space $\mathcal{H}$. For the purposes of this work, these formulations can be assumed to be equivalent. The main advantage of the linear functional notation as used here is that measurements and observables can be kept as abstract symbols and that a mathematically delicate iteration over all possible realizations in all feasible Hilbert spaces can be avoided.   
The measurement data in an experiment described in this model leads to conditional probability distributions, also known as behaviors, 
\begin{align}
    p^s(\vec a, \vec b| \vec x , \vec y) =\omega^s \left( M^0_{a_0|x_0} \dots  N^n_{b_n|y_n} \right),
\end{align}
where we grouped indices as $\vec a= (a_0\dots a_n) $ and so on. 

Let $\omega^s_K$ denote the restriction of a state to a subsystem $K$, i.e. the marginal on $K$. 
We express the independence of a local system of the switch settings by marginal constraints 
\begin{align}
    \omega^s_K=\omega^{s'}_K \quad \forall\; s,s'\in\{0,\dots,n\} 
\end{align}
Intuitively, this expresses that the setting of a switch can not have any influence recognizable by any experiment that only acts on the $K$ system.  
We say that a state $\omega^s$ is $A$--$B$ separable if it can be decomposed as $\omega=\sum p_i w_{i,A_0\dots A_n}w_{i,B_0\dots B_n}$, i.e. as convex hull of product states.
We will use the convention that the switch setting $s=(0,0)$ will be used for key generation. The key will be based on the measurement data obtained by the parties $A_0 B_0$.
Based on this convention we can define correlations with a local model in the key generation round via 
\begin{align}\label{localquantum}
    LQC_{A_0B_0}=\{  p^s(\vec a, \vec b| \vec x , \vec y) \ \vert \   \omega^{(0,0)} \text{is $A$--$B$ separable}, \omega^s_{A_0}=\omega^{s'}_{A_0}, \omega^s_{B_0}=\omega^{s'}_{B_0} \forall s,s'\}
\end{align}
Note that, in a situation  without further test parties $A_1\dots$ and $B_1,\dots$ this definition recovers the well known set of correlations arising from local hidden variable models. Furthermore, in the case of  Bob being  characterized by a concrete Hilbert space, we get the set of behaviors measurable on an unsteerable state. In this sense this definition can be seen a natural extension of these concepts to a situation with partial information obtained by self-testing.  

We can adapt this definition to the 2+1 party setting of Lobo et al. \cite{Lobo2024} (see \autoref{fig:settings}(a)). It is easy to see that our definition recovers the previous definition of short range quantum correlations from \cite{Lobo2024}. A formal proof is given in the appendix \autoref{appendix:eqv} .

We will briefly sketch the main mechanisms of this proof. Followig the nomenclature of \cite{Lobo2024}, we have two states $\omega^{long}$ and $\omega^{short}$, where our party $A_1$ is equivalent to the party $B_S$. $B_L$ corresponds to $B_0$ in our setting. 
Here, our definition of A-B local correlations \autoref{localquantum} reduces to the set of correlations generated by a state $\omega^{long}$ that is separable with respect to $B_L$ and consistent with the  statistics observed by $A-B_S$ on a state $\omega^{short}$. 
In contrast, and after applying the equivalence shown in the previous section, Lobo et al. (see \cite{Lobo2024} Eq. 15)  define the set of short range correlations as all correlations that arise from a joint measurement by $B_L$ on $\omega^{long}$. The equivalence of the two definitions reduces to the claim  that in device-independent models joint measurability and separability are interchangeable concepts at the level of correlations.  
In one direction it has to be shown that any behavior measured on a separable state can be obtained from a joint measurement. As every separable state admits a local hidden variable model, this joint measurement can be directly constructed from the classical hidden variable model (see for example \cite{van2024schmidt} for constructions). 

In the other direction, we have to convince ourselves that any behavior measured by a joint measurement can be obtained from a separable state. W.l.o.g. any fixed joint measurement can be represented by a projective parent measurement on a larger space by Naimark's dilation Theorem. These projections generate an abelian subalgebra of $B_{long}$, which is large enough to describe the specific behavior. As abelian algebras cannot be entangled with any other system the corresponding state must be separable \cite{van2024schmidt}.  

\section{Robust lifts of DIQKD to device-dependent QKD}\label{sec:robust_lifts}

It has been recognized in earlier work regarding routed Bell setups that routed Bell tests are, in principle, able to lift device-independent models to device-dependent models in the setting of a perfect local test as part of a routed BB84 protocol \cite{Lim2013,Chaturvedi2024,Lobo2024,Le_Roy_Deloison_2025}. In this specific setting of a routed BB84 protocol Alice and Bob perform BB84 measurements and Alice does a local CHSH test in order to certify her device. It is well-known that in this case the optimal BB84 key-rate, known as Shor-Preskill key-rate \cite{Shor_2000} is achieved as shown in \cite{Woodhead_2016}. However, hitherto a concrete and general mathematical modeling taking into account an error due to imperfect local tests was still missing, as it would require a smooth interpolation between device-independent and device-dependent setups. These, however, are fundamentally different from the perspective of assumptions. This section is dedicated to a rigorous mathematical analysis of this observation using the recently introduced notion of \emph{robust self-tests} \cite{operator_algebra_self_test2024}. In order to introduce the language of robust self-tests, we will start by treating bipartite quantum experiments between Alice and Bob where both parties query their boxes with local private randomness. Each of them has a finite input set, called $X$ for Alice and $Y^\prime$ for Bob, and a finite output set denoted by $A$ for Alice and $B$ for Bob. We assume a setup in which Alice and Bob are spacelike separated. Hence, Alice and Bob will have in the asymptotic limit a conditional probability distribution
\begin{align}\label{eq:general_conditional_probability}
    p(a,b\vert x,y), \qquad a \in A, \ b \in B, \ x \in X, \ y\in Y^\prime, 
\end{align}
which satisfies the usual non-signaling conditions imposed by the spacelike separation. A \emph{self-test} of a conditional probability distribution $p(a,b\vert x,y)$ is, in a sense, a unique quantum model defined by a tuple 
\begin{align}\label{eq:def_quantum_model}
    \tilde{S} &\coloneqq   \left(\mathcal{H}_{\tilde{A}}, \mathcal{H}_{\tilde{B}}, \left\{\tilde{M}_{ a\vert x}\right\},\left\{\tilde{N}_{b\vert y}\right\},\ket{\varphi_{\tilde{A}\tilde{B}}}\right),
\end{align}
whereby we have fixed (finite-dimensional) Hilbert spaces, fixed measurement operators on these Hilbert spaces and, without loss of generality, a shared pure state. An application of  Born's rule yields
\begin{align}\label{eq:Born_rule}
    p(a,b\vert x,y) = \tr\left[\tilde{M}_{a\vert x} \otimes \tilde{N}_{b\vert y}\ketbra{\varphi_{\tilde{A}\tilde{B}}}{\varphi_{\tilde{A}\tilde{B}}}\right]
\end{align} and each other model of the form in \autoref{eq:def_quantum_model} with the same statistics is connected to \autoref{eq:def_quantum_model} via isometries and auxiliary states \cite{Paddock_2023}. One then defines the unique model via the smallest model ordered by isometries mapping between them (in the finite-dimensional case, this is just the model with the smallest dimensions of $\mathcal{H}_{\tilde{A}}$ and $\mathcal{H}_{\tilde{B}}$). Setups which allow for self-testing are exceptional, because usually completely different models can lead to the same statistics as in \autoref{eq:general_conditional_probability}. 

In DIQKD we are interested in the randomness we can extract purely from the statistics produced by a quantum experiment,  without making any further assumptions on the devices. The experiment is described by a general but unknown quantum model 
\begin{align}\label{eq:general_quantum_model}
    S \coloneqq   \left(\rho_{AB},\{M_{a\vert x}\},\{N_{b\vert y}\},\mathcal{H}_A,\mathcal{H}_B\right). 
\end{align} 
corresponding to the observed statistics via Born's rule (see \autoref{eq:Born_rule}). Importantly, $\rho_{AB}$ does not need to be pure.
Generally, given any statistic $p$, in DIQKD we are interested in an optimization over all compatible models 
\begin{equation}\label{eq:optimization_DIQKD}
    \begin{aligned}
        \inf \ &H(A\vert X=\tilde{x},E) \\
        \operatorname{s.t.} \ &S \quad \text{quantum model} \\
        &S \quad \text{yields the statistics} \ p.
    \end{aligned}
\end{equation}
In general, the optimization problem in \autoref{eq:optimization_DIQKD} is very difficult and few general techniques exist, which rely on techniques from non-commutative optimization \cite{Brown_2024,kossmann2025reliableentropyestimationobserved}. These techniques scale badly with the problem size. However, in the situation of an exact self-test, the conditional von Neumann entropy $H(A\vert X=\tilde{x},E)$ for a specific measurement setting can be calculated from the unique model in \autoref{eq:def_quantum_model} as sketched in the following proposition. Thus, the optimization in \autoref{eq:optimization_DIQKD} simplifies drastically.

\begin{proposition}\label{prop:perfect_self_test_conditional_vonNeumann}
    Let $\left(\mathcal{H}_A,\mathcal{H}_B,\{M_{a\vert x}\},\{N_{b\vert y}\},\rho_{AB}\right)$ be a realization of the statistics $p(a,b\vert x,y)$ and assume that there exists an exact self-test for $p(a,b\vert x,y)$ given by $\left(\mathcal{H}_{\tilde{A}},\mathcal{H}_{\tilde{B}},\{\tilde{M}_{a\vert x}\},\{\tilde{N}_{b\vert y}\},\ket{\psi_{\tilde{A}\tilde{B}}}\right)$. Then we have 
    \begin{align}\label{eq:conditional_vonNeumann}
        H(A\vert X=\tilde{x},E)_{\rho}
        =
        H(A\vert X=\tilde{x})_{\ket{\psi_{\tilde{A}\tilde{B}}}}
        =
        H((q_a)_{a\in A}),
    \end{align}
    where
    \begin{align}
        q_a \coloneqq \tr\!\left[\ketbra{\psi_{\tilde{A}\tilde{B}}}{\psi_{\tilde{A}\tilde{B}}}\,
        \tilde{M}_{a\vert \tilde{x}}\otimes \mathds{1}_{\tilde{B}}\right].
    \end{align}
\end{proposition}

\begin{proof}
    By the exact self-test property, there exist local isometries
    \begin{align}
        I_A: \mathcal{H}_A \to \mathcal{H}_{\tilde{A}} \otimes \mathcal{H}_{\operatorname{aux}}^A
        \quad \text{and} \quad
        I_B: \mathcal{H}_B \to \mathcal{H}_{\tilde{B}} \otimes \mathcal{H}_{\operatorname{aux}}^B,
    \end{align}
    such that for a purification $\ket{\psi_{ABE}}$ of $\rho_{AB}$ we have
    \begin{align}
        (I_A\otimes I_B\otimes \mathds{1}_E)\ket{\psi_{ABE}}
        =
        \ket{\Phi_{\operatorname{aux}_A \operatorname{aux}_B E}}
        \otimes
        \ket{\psi_{\tilde{A}\tilde{B}}}.
    \end{align}
    Hence, after measuring $A$ with setting $\tilde{x}$, the resulting classical-quantum state factorizes as
    \begin{align}
        \rho_{AE}^{\tilde{x}}
        \;\cong\;
        \ketbra{\Phi_{\operatorname{aux}_A \operatorname{aux}_B E}}{\Phi_{\operatorname{aux}_A \operatorname{aux}_B E}}_{\operatorname{aux}_A \operatorname{aux}_B E}
        \otimes
        \sum_a q_a \ketbra{a}{a}_A,
    \end{align}
    up to the above local isometries. Since conditional von Neumann entropy is invariant under local isometries on the conditioning system and unchanged by tensoring a product state on that system, it follows that
    \begin{align}
        H(A\vert X=\tilde{x},E)_\rho
        =
        H(A\vert X=\tilde{x})_{\ket{\psi_{\tilde{A}\tilde{B}}}}.
    \end{align}
    Finally, on the ideal state the register $A$ is classical with distribution $(q_a)_a$, so
    \begin{align}
        H(A\vert X=\tilde{x})_{\ket{\psi_{\tilde{A}\tilde{B}}}}
        =
        H((q_a)_{a\in A}).
    \end{align}
\end{proof}

The key difference between \autoref{prop:perfect_self_test_conditional_vonNeumann} and a local self-test in a routed setup is that the Alice-Bob correlations are used not only for key generation, but also for self-testing the devices. This leads to two natural questions.

First, what is the effect of noise in the observed statistics, and how does the conditional von Neumann entropy behave under such perturbations? In general, inferring the conditional von Neumann entropy (defined via a purification) directly from observed statistics is difficult, which is also reflected by the optimization over all compatible models being a hard problem. For self-tests, the notion of robustness introduced in \cite{operator_algebra_self_test2024} provides the appropriate framework, and one goal of this work is to apply it to the conditional von Neumann entropy. Generally, it is unclear how the conditional von Neumann entropy behaves for certain statistics if one varies the statistics slightly. This phenomenon has been observed in \cite{Farkas2021} and \cite{Wooltorton2024,Farkas2024}, where for arbitrarily small non-locality, quantified in terms of Bell-violation, both statistics allowing for secret key extraction and not allowing for secret key extraction have been found. Here, it is obvious that knowing the statistics (or slightly more general Bell-violations) is not sufficient to infer something about the conditional von Neumann entropy.   

Second, self-tests can also certify the devices themselves, i.e. the measurement operators. In routed Bell-test scenarios, noise between distant parties is most likely governed by the physical Lambert-Beer effect. This suggests that one may self-test the devices locally and interpret the Lambert--Beer effect as effectively reducing the bit-error rate in key-generating rounds.

To connect with the results in \cite{Tomamichel_2013,Lim2013,Chaturvedi2024,Lobo2024,Le_Roy_Deloison_2025}, we proceed in two steps. First, we adapt the framework of robust self-tests to the BB84 setting, recovering the well-known conclusions of \cite{Tomamichel_2013,Lim2013}. Second, we extend the analysis to the setting of general routed Bell tests with self-test.

\subsection{Routed BB84 protocol}\label{subsec:bb84_protocol}

The simplest model to investigate many of the interesting properties of a routed Bell setup is the routed BB84 protocol. The routed BB84 protocol is defined by the actions of three parties: Alice, Bob, and Fred. Each of them has two projections, and Alice and Bob measure noisy BB84 statistics, while Alice and Fred measure CHSH statistics. It is well known that the entropic uncertainty relation under side information \cite{Berta2010} is tight for the BB84 protocol, which is also used in the semi-DI framework \cite{Woodhead_2016}. Semi-DI means that one party, usually Alice, is fully trusted and a concrete Hilbert space representation of her device is known. The step from semi-DI to routed BB84 is simply to allow noise in the local CHSH test, because, in fact, an optimal local CHSH test self-tests Alice's device perfectly, such that we would be back in the semi-DI case \cite{Mayers2004,Paddock_2023}. This observation was made in \cite{Le_Roy_Deloison_2025} and verified numerically in \cite{kossmann2025routedbelltestsarbitrarily}. Without yet stating what robustness in self-testing means, this already shows, in some sense, a robustness effect of the key rate depending on the quality of the local CHSH test. The key for an analysis of the more refined routed BB84 protocol with dependency on the local test is the adaptive uncertainty relation from \cite{Tomamichel_2013}, which yields tight bounds in the case of an optimal CHSH test and can be reduced to the observation that, in that case, the whole analysis depends only on the overlap between the two projections in Alice's device. Thus, from a theoretical perspective, the analysis of this setup can be reduced to the effective overlap of Alice's two projections. This drastically simplifies the analysis and even yields quantitative bounds on the key rate. Nevertheless, we emphasize and show subsequently in \autoref{thm:robust_self_test_routedBB84} that the theory of robust self-tests can already be applied, as the CHSH test is a robust self-test \cite[Thm.~9.5]{operator_algebra_self_test2024}. The following definition is from \cite{Tomamichel_2013}.

\begin{definition}[Effective overlap]\label{def:effective_overlap}
Let $\rho_A \in \mathcal{S}(\mathcal{H}_A)$ and let
$Z=\{M_{a\mid Z}\}_a$ and $X=\{M_{b\mid X}\}_b$ be POVMs on $A$. The \emph{effective overlap} of $(\rho_A,X,Z)$ is defined as
\begin{align}\label{eq:definition_oevrlap_1}
c^\star(\rho_A,X,Z)
\coloneqq  
\inf_{\left(U,X',Z',K'\right)}
\sum_k \tr\left[P_{A'}^{k}\,U\rho_A U^\dagger\right]\;
\max_{b}\left\lVert
\sum_a
\left(P_{A'}^{k}\Pi_{a\mid Z}P_{A'}^{k}\right)
\left(P_{A'}^{k}\Pi_{b\mid X}P_{A'}^{k}\right)
\left(P_{A'}^{k}\Pi_{a\mid Z}P_{A'}^{k}\right)
\right\rVert.
\end{align}
Here, the infimum ranges over all
isometries $U:\mathcal{H}_A \to \mathcal{H}_{A'}$,
all POVMs $Z'=\{\Pi_{a\mid Z}\}_a$ and $X'=\{\Pi_{b\mid X}\}_b$ on $A'$,
and all projective measurements $K'=\{P_{A'}^k\}_k$ on $A'$,
such that for all outcomes $a,b$,
\begin{align}\label{eq:definition_overlap_2}
\sum_k U^\dagger P_{A'}^k \Pi_{a\mid Z} P_{A'}^k U = M_{a\mid Z}
\quad\text{and}\quad
\sum_k U^\dagger P_{A'}^k \Pi_{b\mid X} P_{A'}^k U = M_{b\mid X}.
\end{align}
\end{definition}
We note that the effective overlap is continuous in its input state.
\begin{lemma}[Continuity of the effective overlap in the marginal state]\label{lem:cstar_continuity_state}
Fix Alice's two measurements $X,Z$.
Then for any two density operators $\tau_A,\tau'_A$ on $\mathcal H_A$,
\begin{align}\label{eq:cstar_lipschitz}
\lvert c^\star(\tau_A,X,Z)-c^\star(\tau'_A,X,Z)\rvert
\ \le\ \frac{1}{2}\lVert \tau_A-\tau^\prime_A\rVert_1.
\end{align}
\end{lemma}
\begin{proof}
    Denote by $\gamma\coloneqq   \left(U,X',Z',K'\right)$ an admissible tuple. Then we define
    \begin{align}
        f_{\gamma}\left(\tau_A\right) \coloneqq   \sum_k \tr\left[P_{A'}^{k,\gamma}\,U_\gamma \tau_A U^\dagger_\gamma \right] c_{k,\gamma}
    \end{align}
    with 
    \begin{align}
        c_{k,\gamma} \coloneqq   \max_{b}\left\lVert
\sum_a
\left(P_{A'}^{k,\gamma}\Pi_{a\mid Z}^{\gamma}P_{A'}^{k,\gamma}\right)
\left(P_{A'}^{k,\gamma}\Pi_{b\mid X}^{\gamma}P_{A'}^{k,\gamma}\right)
\left(P_{A'}^{k,\gamma}\Pi_{a\mid Z}^{\gamma}P_{A'}^{k,\gamma}\right)
\right\rVert.
    \end{align}
Now we can write 
\begin{equation}
\begin{aligned}
    \inf_{\gamma} f_\gamma\left(\tau_A\right) - \inf_{\gamma} f_\gamma\left(\tau_A^\prime\right) &\leq \sup_{\gamma} f_\gamma\left(\tau_A\right) - f_\gamma\left(\tau_A^\prime\right) \\
    &= \sup_{\gamma} \tr\left[\sum_k c_{k,\gamma}U_\gamma^\dagger P_{A^\prime}^{k,\gamma} U_\gamma\left(\tau_A - \tau_A^\prime\right) \right].
\end{aligned}
\end{equation}
But we also have $0\leq \sum_k c_{k,\gamma}U_\gamma^\dagger P_{A^\prime}^{k,\gamma} U_\gamma\leq \mathds{1}_A$, implying the result with the variational characterization of the trace-distance. 
\end{proof}

Adapting the setup from \cite{Tomamichel_2013,Lim2013} to \cite{kossmann2025routedbelltestsarbitrarily} and thus \autoref{fig:settings}~(b) can be done as follows. Let $\omega\left(\operatorname{CHSH}\right) \in [0,2\sqrt{2}]$ denote the observed CHSH value 
between Alice and Fred, which corresponds then to the effective overlap $c^\star(\sigma_A,X_A,Z_A)$. Given the connection between effective overlap and CHSH value found in \cite{Tomamichel_2013}
\begin{align}
c^\star(\sigma_A,X,Z) \leq \frac{1}{2}+\frac{\omega\left(\operatorname{CHSH}\right)}{8}\sqrt{8-\omega\left(\operatorname{CHSH}\right)^2},
\end{align}
we can apply \autoref{lem:cstar_continuity_state} for $\lVert \sigma_A - \rho_A\rVert_1 \leq \varepsilon.$ Then we have
\begin{align}\label{eq:cstar_beta_bound_eps}
c^\star(\rho_A,X,Z) \leq \frac{1}{2}+\frac{\omega\left(\operatorname{CHSH}\right)}{8}\sqrt{8-\omega\left(\operatorname{CHSH}\right)^2} + \frac{\varepsilon}{2}.
\end{align}
Now we apply the uncertainty relation under side-information from \cite{Tomamichel_2013} for the tripartite quantum state $\rho_{ABE}$ and $X_A = \{M_{a\vert Z}\}_a$ and $Z_A = \{M_{b\vert X}\}_b$ measurements such that $X_A,Z_A  \subseteq \mathcal{B}(\mathcal{H}_A)$. Then the corresponding reduced states satisfy
\begin{align}\label{eq:uncertainty_tomamichel}
    H(Z_A \vert E)_\rho + H(X_A\vert B)_\rho \geq - \log c^\star(\rho_A,X_A,Z_A).
\end{align}
By a standard argument and Fano's inequality we conclude for the asymptotic one-way BB84 key rate the explicit bound (cf. for details the proof of \autoref{thm:robust_self_test_routedBB84})
\begin{align}\label{eq:keyrate_beta_explicit_eps}
r_{\mathrm{key}}
\ \ge\
1\,-\,\log_2\left(1+\frac{\omega\left(\operatorname{CHSH}\right)}{4}\sqrt{8-\omega\left(\operatorname{CHSH}\right)^2}+\varepsilon\right)
\;-\; h_2(Q_Z) \;-\; h_2(Q_X).
\end{align}
This example shows many of the interesting properties of a routed Bell experiment. First of all, it shows that the value $\omega\left(\operatorname{CHSH}\right)$ can be incorporated into the key rate and thus demonstrates the smooth dependence between the CHSH value and the BB84 key rate.  
However, comparing this result with the numerical key-rate results for the case of one-sided self-tests and full statistics from \cite{kossmann2025routedbelltestsarbitrarily}, \autoref{eq:keyrate_beta_explicit_eps} yields strictly worse bounds on the asymptotic key rate. 
In the comparison of \cite{Lim2013,Chaturvedi2024,Lobo2024,Le_Roy_Deloison_2025} to \cite{kossmann2025routedbelltestsarbitrarily}, the latter framework, understood as a mathematical model for DIQKD, shows with \autoref{lem:cstar_continuity_state} the direct dependency on the marginal constraint, which is implicitly incorporated in all other frameworks, as \autoref{example:assumptions_switch} shows. Concretely, \autoref{eq:keyrate_beta_explicit_eps} makes explicit how the key rate depends on the condition $\Vert\rho_A-\sigma_A\Vert_1\leq \varepsilon$, and thus shows that key can still be extracted even when the marginal constraint is only approximate.

For the following investigations, it is worth mentioning that \autoref{eq:keyrate_beta_explicit_eps} does not include any information about a potential second self-testing partner on Bob's side. Given that the extractable randomness is drawn from Alice's device, it is also unclear how to include this directly in \autoref{eq:keyrate_beta_explicit_eps}. However, numerical results in \cite{kossmann2025routedbelltestsarbitrarily} indicate that a second self-test on Bob's side yields a strict improvement in the asymptotic key rate. In the following, we develop a framework that includes both switches.

\subsection{General robust self-testing and conditional von Neumann entropies}\label{subsec:general_robust_selftest}

The ideas of \autoref{subsec:bb84_protocol} show that the routed Bell-test framework is, at least in principle, able to derive smooth interpolation results between fully DI scenarios and device-dependent QKD. However, techniques based on uncertainty relations usually depend only on a one-sided analysis, as they concern the extractable randomness from Alice's device. Moreover, uncertainty relations for more than two measurements are inherently hard, and they remain open even in the fully device-dependent case beyond qubits. In this subsection, we introduce the framework of robust self-testing from \cite{operator_algebra_self_test2024} and connect the results to the general theoretical building blocks.  Our first goal is to reprove the results from \autoref{subsec:bb84_protocol} and then to generalize them and to incorporate self-tests with both parties Alice and Bob. A robust self-test is defined as follows.

\begin{definition}[Robust self-test]\label{def:robust_self_test}
    Let $\varepsilon>0$ and let two quantum models be given as
    \begin{equation}
    \begin{aligned}
        S &= \left(\mathcal{H}_A,\mathcal{H}_B, \{M_{a\vert x}\},\{N_{b\vert y}\},\ket{\psi}\right) \}\\
        \tilde{S} &= \left(\mathcal{H}_{\tilde{A}},\tilde{\mathcal{H}}_B, \{\tilde{M}_{a\vert x}\},\{\tilde{N}_{b\vert y}\},\ket{\tilde{\psi}}\right) \}.
    \end{aligned}
    \end{equation}
    We say $\tilde S$ is a local $\varepsilon$-dilation of $S$, denoted as $S\succeq_\varepsilon \tilde S$, if there are isometries $I_A:\mathcal{H}_A \to \mathcal{H}_{\tilde{A}} \otimes \mathcal{H}_A^{\operatorname{aux}}$ and $I_B:\mathcal{H}_B \to \tilde{\mathcal{H}}_B \otimes \mathcal{H}_B^{\operatorname{aux}}$ and a vector state $\ket{\operatorname{aux}} \in \mathcal{H}_A^{\operatorname{aux}} \otimes \mathcal{H}_B^{\operatorname{aux}}$, such that 
    \begin{equation}
        \begin{aligned}
            \Vert I_A \otimes I_B \left(M_{a\vert x}\otimes \mathds  1_B \ket \psi\right) - \left(\tilde{M}_{a\vert x} \otimes \mathds 1_B \ket{\tilde \psi}\right)\otimes \ket{\operatorname{aux}}\Vert_2 \leq \varepsilon \\
            \Vert I_A \otimes I_B \left(\mathds  1_A\otimes N_{b\vert y} \ket \psi\right) - \left( \mathds 1_A \otimes \tilde{N}_{b\vert y}  \ket{\tilde \psi}\right)\otimes \ket{\operatorname{aux}}\Vert_2 \leq \varepsilon \\
            \Vert I_A\otimes I_B \ket{\psi} - \ket{\tilde\psi} \otimes \ket{\operatorname{aux}}\Vert_2\leq \varepsilon
        \end{aligned}
    \end{equation}
\end{definition}

The following theorem shows the result from \autoref{eq:keyrate_beta_explicit_eps} as a qualitative result in the language of robust self-tests.
\begin{theorem}\label{thm:robust_self_test_routedBB84}
In the routed BB84 protocol with $\sigma_{AA_F}$ a self-test for Alice's device performing a CHSH game with winning probability $\omega(\mathrm{CHSH})-\varepsilon$ and Alice and Bob using $\rho_{AB}$ to run a BB84 protocol, the one-way communication, device-independent key-rate satisfies
\begin{align}
r_{\mathrm{key}} \geq 1 - h(Q_Z) - h(Q_X) - \mathcal{O}(\sqrt{\varepsilon}) \, .
\end{align}
In particular, as $\varepsilon\to 0$, they attain the optimal BB84 key-rate \cite{Shor_2000}.
\end{theorem}
\begin{proof}
We assume that the winning probability $\omega(\operatorname{CHSH})-\varepsilon$ is achieved by Alice and Fred given by the (w.l.o.g.) PVM model
\begin{align}
S \coloneqq   (\mathcal{H}_A,\mathcal{H}_F, \{M_{a\vert x}\},\{N_{b\vert y}\},\ket{\psi}) 
\end{align}
and we call the reduced states of that model $\sigma_A$ and $\sigma_F$. Moreover, we call the optimal CHSH model on a qubit
\begin{align}
\tilde{S} \coloneqq   (\mathcal{H}_{\tilde{A}},\tilde{\mathcal{H}}_F, \{\tilde{M}_{a\vert x}\},\{\tilde{N}_{b\vert y}\},\ket{\tilde{\psi}})
\end{align}
with $\tilde{M}_{a\vert X}$ the optimal Pauli-$x$ measurements, $\tilde{M}_{a\vert Z}$ the optimal Pauli-$z$ measurements on $\mathcal{H}_{\tilde{A}}$, and $\ket{\tilde{\psi}}$ a maximally entangled state on two qubits. By Theorem~9.5 in~\cite{operator_algebra_self_test2024}, achieving $\omega(\operatorname{CHSH})-\varepsilon$ is a robust self-test, i.e.\ there exists $\delta=O(\sqrt{\varepsilon})$ such that $\tilde S$ is centrally supported and
\begin{align}
S\succeq_{\delta}\tilde S
\end{align}
via local isometries $I_A,I_F$ and an auxiliary state $\ket{\operatorname{aux}}$.

Let
\begin{align}
\ket{\Psi'}\coloneqq   (I_A\otimes I_F)\ket{\psi},
\qquad
\ket{\phi}\coloneqq   \ket{\tilde\psi}\otimes\ket{\operatorname{aux}},
\qquad
\zeta\coloneqq   \Vert\ket{\Psi'}-\ket{\phi}\Vert_2 = O(\delta).
\end{align}
Define Alice's reflections in $S$ by
\begin{align}
X \coloneqq   M_{0\vert X}- M_{1\vert X}, \qquad Z \coloneqq   M_{0\vert Z}- M_{1\vert Z}.
\end{align}
Fix the isometry $I_A$ from the self-test and consider the associated dilated projective measurements $\{\Pi_{a\mid x}\}$, so that we can set the (dilated) reflections
\begin{align}
X^\prime \coloneqq   \Pi_{0\vert X} - \Pi_{1\vert X}, \qquad
Z^\prime \coloneqq   \Pi_{0\vert Z} - \Pi_{1\vert Z},
\qquad
\sigma_A^{\prime} \coloneqq   I_A \sigma_A I_A^\dagger .
\end{align}
Let $\tilde X=\sigma_x$ and $\tilde Z=\sigma_z$ on the ideal qubit.

From \cite[Prop.~3.15]{operator_algebra_self_test2024} with $k=1$ (and linearity to pass from projectors to reflections), the robust self-test yields the \emph{on-state} bounds
\begin{align}\label{eq:on-state-bnds-proof-fixed}
    \eta_X \coloneqq   \big\Vert (X' - \tilde X\otimes \mathds 1)\ket{\phi}\big\Vert_2 = O(\delta),
    \qquad
    \eta_Z \coloneqq   \big\Vert (Z' - \tilde Z\otimes \mathds 1)\ket{\phi}\big\Vert_2 = O(\delta).
\end{align}
We also need control of the \emph{products}. Note that $X'Z'$ (and $Z'X'$) is a linear combination of length-$2$ monomials in the POVM generators (expand $(M_{0\mid X}-M_{1\mid X})(M_{0\mid Z}-M_{1\mid Z})$ and use linearity), hence \cite[Prop.~3.15]{operator_algebra_self_test2024} with $k=2$ yields
\begin{align}\label{eq:product-bnds}
    \kappa_{XZ} \coloneqq   \big\Vert (X'Z' - \tilde X\tilde Z\otimes \mathds 1)\ket{\phi}\big\Vert_2 = O(\delta),
    \qquad
    \kappa_{ZX} \coloneqq   \big\Vert (Z'X' - \tilde Z\tilde X\otimes \mathds 1)\ket{\phi}\big\Vert_2 = O(\delta).
\end{align}
Since $\{\tilde X,\tilde Z\}=0$, we can write on $\ket{\phi}$,
\begin{align}
\{X',Z'\}\ket{\phi}
=
\bigl(X'Z'-\tilde X\tilde Z\otimes \mathds 1\bigr)\ket{\phi}
+
\bigl(Z'X'-\tilde Z\tilde X\otimes \mathds 1\bigr)\ket{\phi},
\end{align}
and therefore, by the triangle inequality and \autoref{eq:product-bnds},
\begin{align}\label{eq:vector-bound-phi-fixed}
\big\Vert\{X',Z'\}\ket{\phi}\big\Vert_2
\le \kappa_{XZ}+\kappa_{ZX}
= O(\delta).
\end{align}
Passing to $\ket{\Psi'}$ by the triangle inequality and using $\Vert\{X',Z'\}\Vert\le 2$ (since $X',Z'$ are reflections),
\begin{align}\label{eq:vector-bound-psi-fixed}
    \big\Vert\{X',Z'\}\ket{\Psi'}\big\Vert_2 
    \le \big\Vert\{X',Z'\}\ket{\phi}\big\Vert_2 + \Vert\{X',Z'\}\Vert\,\zeta
    \le O(\delta) + 2\,O(\delta)
    = O(\delta).
\end{align}
Since
\begin{align}
\tr[\sigma_A' \{X',Z'\}^2]
=
\langle \Psi'|\{X',Z'\}^2|\Psi'\rangle
=
\Vert\{X',Z'\}\ket{\Psi'}\Vert_2^2,
\end{align}
\autoref{eq:vector-bound-psi-fixed} implies the quadratic bound
\begin{align}\label{eq:second-moment-quadratic-fixed}
    \tr\left[\sigma_A^\prime \{X^\prime,Z^\prime\}^2 \right] = O(\delta^2).
\end{align}
Combining \autoref{appendix:qubit_reduction} 
with \autoref{eq:second-moment-quadratic-fixed} and $\delta=O(\sqrt{\varepsilon})$ gives
\begin{align}
c^\star(\sigma_A,X,Z) \le \frac{1}{2} + O(\sqrt{\varepsilon}).
\end{align}

Now we use the protocol assumption that $\sigma_A = \rho_A$ in order to pass to the uncertainty relation from \cite{Tomamichel_2013} (for a purification of $\rho_{AB}$, repeated in \autoref{eq:uncertainty_tomamichel}):
\begin{equation}
\label{eq:EUR_final}
H(Z_A|E)_\rho + H(X_A|B)_\rho \geq -\log c^\star(\rho_A,X,Z)
\geq -\log \frac{1}{2} + \mathcal{O}(\sqrt{\varepsilon}) .
\end{equation}
By data-processing, Fano’s inequality and the observed bit-error rates, we have
\begin{align}\label{eq:thm_fano_inequality_proof}
H(X_A|B)\leq H(X_A|X_B) \leq h(Q_X),\qquad
H(Z_A|Z_B) \leq h(Q_Z).
\end{align}
Inserting the ingredients into the Devetak-Winter formula with key extracted from $Z$ gives
\begin{align}
r_{\operatorname{key}} \geq H(Z_A|E) - H(Z_A|Z_B).
\end{align}
From \autoref{eq:EUR_final} we have $H(Z_A|E)\geq -\log c^\star(\rho_A,X,Z) - H(X_A|B)$, hence using $-\log c^\star(\rho_A,X,Z)=1-\mathcal{O}(\sqrt{\varepsilon})$,
\begin{align}
r_{\operatorname{key}}
&\geq \bigl[1-\mathcal{O}(\sqrt{\varepsilon}) - H(X_A|B)\bigr] - H(Z_A|Z_B) \\
&\geq 1-\mathcal{O}(\sqrt{\varepsilon}) - h(Q_X)- h(Q_Z),
\end{align}
which yields in the limit $\varepsilon\to 0$ the optimal asymptotic BB84 rate $r_{\operatorname{key}}\geq 1-h(Q_X)-h(Q_Z)$.
\end{proof}

\subsection{General lifts of device-independent tasks}

In this section we aim to investigate \autoref{fig:two_switches_lift} from a general perspective of security proofs and rate functions and how this setup could lead to a framework for device-independent security proofs. As the notion of \emph{robust self-tests} from \cite{operator_algebra_self_test2024} is quite technical, we make the following general assumptions holding for the rest of this section.
\begin{assumptions}\label{assumptions}
We assume in the following to have a general model for sets $X,Y^\prime$
\begin{align}
    \left(\mathcal{H}_A, \mathcal{H}_B, \{M_{ a\vert x}\},\{N_{b\vert y}\},\rho_{AB}\right)
    \tag{GM}\label{eq:GM}
\end{align}
and two $\varepsilon_A$- and $\varepsilon_B$-dilations for the set of measurement operators $X$
\begin{align}
    S_A \coloneqq    \left(\mathcal{H}_A, \mathcal{H}_F, \{M_{ a\vert x}\},\{P_{p\vert q}\},\ket{\psi_{AF}} \right) \quad \text{of} \quad  \tilde{S}_A \coloneqq    \left(\mathcal{H}_{\tilde{A}}, \mathcal{H}_{\tilde{F}}, \{\tilde{M}_{ a\vert x}\},\{\tilde{P}_{p\vert q}\},\ket{\varphi_{\tilde{A}\tilde{F}}} \right)
    \tag{self-A}\label{eq:selfA}
\end{align}
that is $S_A \succ_{\varepsilon_A} \tilde{S}_A$ and similarly for a subset $Y\subseteq Y^\prime$
\begin{align}
    S_B \coloneqq    \left(\mathcal{H}_B, \mathcal{H}_G, \{N_{b\vert y}\},\{Q_{r\vert s}\},\ket{\mu_{BG}} \right) \quad \text{of} \quad  \tilde{S}_B \coloneqq    \left(\mathcal{H}_{\tilde{B}}, \mathcal{H}_{\tilde{G}},\{\tilde{N}_{b\vert y}\},\{\tilde{Q}_{r\vert s}\} ,\ket{\nu_{\tilde{A}\tilde{F}}} \right)
    \tag{self-B}\label{eq:selfB}
\end{align}
that is $S_B \succ_{\varepsilon_B} \tilde{S}_B$. Furthermore assume
\begin{align}
    \rho_A = \psi_A \qquad \text{and} \qquad \rho_B = \mu_B.
    \tag{MC}\label{eq:MC}
\end{align}
From the $\varepsilon_A$ and $\varepsilon_B$ dilations there exist by definition isometries $I_A$ and $I_B$ respectively. Define 
\begin{equation}
\tilde{\rho}_{\tilde A\tilde B}
\coloneqq   
\tr_{\operatorname{aux}}\left[\left(I_A\otimes I_B\right)\,\rho_{AB}\, \left(I_A\otimes I_B\right)^\dagger\right]
\, \in \, \mathcal S(\tilde{\mathcal H}_A\otimes \tilde{\mathcal H}_B).
\tag{RS}\label{eq:RS}
\end{equation}
In the following we moreover aim to generalize \autoref{eq:optimization_DIQKD} to more general conditional entropies as ,e.g., defined in \cite[Def.~1]{hahn2025analyticrenyientropybounds}, that is let $\mathbb{H}$ be a general conditional entropy and a given distribution $p(a,b\vert x,y)$, $a \in A,b\in B,x\in X,y\in Y$. Then we are interested in
\begin{equation}
\begin{aligned}
\operatorname{OPT}_{\mathrm{DI}}^{\mathbb{H}}(p;\tilde x)
\coloneqq   \inf\ 
& \mathbb{H}\left(A_{\tilde x}\vert E\right)_{(\mathcal M_{\tilde x}\otimes \id_E)(\rho_{AE})} \\
\operatorname{s.t.}\
& \tr\left[\rho_{AB}\,M_{a|x}\otimes N_{b|y}\right] = p_{abxy},\quad  a,b,x,y,\\
& \rho_{ABE}=\ketbra{\psi}{\psi},\ \{M_{a|x}\},\{N_{b|y}\}\ \text{projective}.
\end{aligned}
\tag{OPT-DI}\label{eq:OPT-DI}
\end{equation}
\end{assumptions}
In the following we show that \autoref{assumptions} implies that routed Bell-tests can be seen as a technique for self-testing the underlying local devices in Alice's and Bob's lab for key generation.  

\begin{figure}
    \centering
    \includegraphics[width=0.8\linewidth]{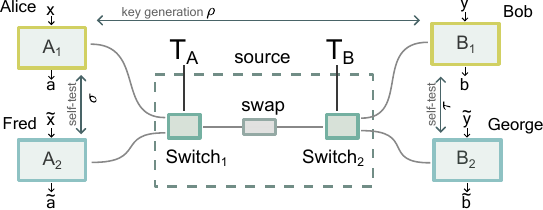}
    \caption{Adapted from the setting in \cite{kossmann2025routedbelltestsarbitrarily}. A routed Bell experiment with four parties. The self-tests are performed by the pairs Alice-Fred and Bob-George. 
    }
    \label{fig:two_switches_lift}
\end{figure}
In \autoref{fig:two_switches_lift} the model with four parties, self-testing the devices from Alice and Bob is shown. 
Our first result shows that, given a robust self-test, we can lift in a four-partite experiment the statistics made to the level of the perfect model. 

Given statistics that admit a robust self-test, the following lemma shows that they can also be realized, up to a small error, by a compressed state in the ideal model. This suggests that, at the level of observed statistics, a device-independent problem can effectively be reduced to a device-dependent one, again up to a small error. As a next step, we show an analogous statement for conditional entropy in a single-round DIQKD setting with locally robust self-tests.
\begin{lemma}[Compression of correlations under local $\varepsilon$-dilations]\label{lem:compression_correlations}
Given the \autoref{assumptions}, then for all $a,b,x,y$,
\begin{equation}\label{eq:probability_error}
\lvert
\tr\left[\rho_{AB}\,M_{a|x}\otimes N_{b|y}\right]
-
\tr\left[\tilde\rho_{\tilde A\tilde B}\,\tilde M_{a|x}\otimes \tilde N_{b|y}\right]
\rvert
\leq 3 \left(\varepsilon_A + \varepsilon_B\right).
\end{equation}
\end{lemma}
\begin{proof}
    Given that we have $\varepsilon_A$-, $\varepsilon_B$-dilations for Alice and Bob at our disposal, we fix arbitrary inputs and outputs $a,b,x,y$ and abbreviate
\begin{align}
    M \coloneqq   M_{a|x},\qquad \tilde M \coloneqq   \tilde M_{a|x},\qquad 
    N \coloneqq   N_{b|y},\qquad \tilde N \coloneqq   \tilde N_{b|y}.
\end{align}
Then, by definition we let the isometries coming from $S_A\succ_{\varepsilon_A}\tilde S_A$ be
\begin{align}
I_A:\mathcal H_A\to \mathcal H_{\tilde A}\otimes \mathcal H_{\operatorname{aux}}^A,
\qquad 
I_F:\mathcal H_F\to \mathcal H_{\tilde F}\otimes \mathcal H_{\operatorname{aux}}^F,
\end{align}
together with some auxiliary unit vector $\ket{\operatorname{aux}_A}\in \mathcal H_{\operatorname{aux}}^A\otimes \mathcal H_{\operatorname{aux}}^F$.
Similarly, from $S_B\succ_{\varepsilon_B}\tilde S_B$ let
\begin{align}
I_B:\mathcal H_B\to \mathcal H_{\tilde B}\otimes \mathcal H_{\operatorname{aux}}^B,
\qquad 
I_G:\mathcal H_G\to \mathcal H_{\tilde G}\otimes \mathcal H_{\operatorname{aux}}^G,
\end{align}
and an auxiliary unit vector $\ket{\operatorname{aux}_B}\in \mathcal H_{\operatorname{aux}}^B\otimes \mathcal H_{\operatorname{aux}}^G$. Let $\mathcal H_{\operatorname{aux}}\coloneqq   \mathcal H_{\operatorname{aux}}^A\otimes \mathcal H_{\operatorname{aux}}^B$. Define the unital completely positive maps 
\begin{align}
\Lambda_A^*(X) \coloneqq   I_A^\dagger\,\left(X\otimes \mathds{1}_{\operatorname{aux}}^A\right)\,I_A,\qquad
\Lambda_B^*(Y)\coloneqq    I_B^\dagger\,\left(Y\otimes \mathds{1}_{\operatorname{aux}}^B\right)\,I_B .
\end{align}
Then, using the definition \autoref{eq:RS} and cyclicity of the trace,
\begin{align*}
\tr\left[\tilde\rho_{\tilde{A}\tilde{B}}\,\tilde{M}\otimes \tilde{N}\right]
&=
\tr\left[\left (I_A\otimes I_B\right)\rho_{AB}\left(I_A\otimes I_B\right)^\dagger \,\left(\tilde{M}\otimes \tilde{N}\otimes \mathds{1}_{\operatorname{aux}}\right)\right]\\
&=
\tr\left[\rho_{AB}\,\left(I_A^\dagger\left(\tilde{M}\otimes \mathds{1}_{\operatorname{aux}}^A\right)I_A\right)\otimes \left(I_B^\dagger\left(\tilde N\otimes \mathds 1_{\operatorname{aux}}^B\right)I_B\right)\right]\\
&= \tr\left[\rho_{AB}\,\Lambda_A^*(\tilde M)\otimes \Lambda_B^*(\tilde N)\right].
\end{align*}
Hence we can conclude
\begin{align}
\lvert \tr\left[\rho_{AB}\,M\otimes N\right]-\tr\left[\tilde{\rho}_{\tilde{A}\tilde{B}}\,\tilde{M}\otimes \tilde{N}\right]\rvert
=
\lvert \tr\left[\rho_{AB}\,M\otimes N\right]-\tr\left[\rho_{AB}\,\Lambda_A^*(\tilde{M})\otimes \Lambda_B^*(\tilde{N})\right]\rvert.
\end{align}
Then we have
\begin{align}\label{eq:split_error}
\lvert \tr\left[\rho_{AB}M\otimes N\right]-\tr[\rho_{AB}\Lambda_A^*(\tilde{M})\otimes \Lambda_B^*(\tilde{N})]\rvert
&\leq
\rvert\tr\left[\rho_{AB}\left(M-\Lambda_A^*\left(\tilde{M}\right)\right)\otimes N\right]\rvert
\\
&\quad+
\rvert \tr\left[\rho_{AB}\Lambda_A^*\left(\tilde M\right)\otimes \left(N-\Lambda_B^*\left(\tilde N\right)\right)\right]\rvert.
\end{align}

Set
\begin{align}
J_A \coloneqq   I_A\otimes I_F:\ \mathcal{H}_A\otimes \mathcal{H}_F \to
(\mathcal{H}_{\tilde A}\otimes \mathcal{H}_{\operatorname{aux}}^A)\otimes (\mathcal{H}_{\tilde F}\otimes \mathcal{H}_{\operatorname{aux}}^F).
\end{align}
We consider now the following assertion
\begin{align} \label{eq:3-eps-bound-on-A}
    \lVert \left(M-I_A^\dagger \tilde{M} \otimes \mathds{1}_{\operatorname{aux}}^A I_A \otimes \mathds{1}_F \right)\ket{\psi_{AF}} \rVert_2 \leq 3\varepsilon
\end{align}

Indeed, since $J_A$ is an isometry, it preserves the $2$-norm, hence
\begin{align}
\lVert \left(M-I_A^\dagger \tilde{M} \otimes \mathds{1}_{\operatorname{aux}}^A I_A \right)\otimes \mathds{1}_F\ket{\psi_{AF}}\rVert_2
=\lVert J_A \left(M-I_A^\dagger \tilde{M} \otimes \mathds{1}_{\operatorname{aux}}^A I_A \right)\otimes \mathds{1}_F\ket{\psi_{AF}}\rVert_2.
\end{align}

We compute $J_A\left(M-I_A^\dagger \tilde{M} \otimes \mathds{1}_{\operatorname{aux}}^A I_A \right)\otimes \mathds{1}_F\ket{\psi_{AF}}$ explicitly:
\begin{align*}
J_A\left(I_A^\dagger \tilde{M} \otimes \mathds{1}_{\operatorname{aux}}^A I_A \right)\otimes \mathds{1}_F\ket{\psi_{AF}}
&=
(I_A\otimes I_F)\left(I_A^\dagger \tilde{M} \otimes \mathds{1}_{\operatorname{aux}}^A I_A \right)\otimes \mathds{1}_F\ket{\psi_{AF}} \\
&=
\left(\left(I_A I_A^\dagger\right)\left(\tilde{M}\otimes \mathds{1}_{\operatorname{aux}}^A\right)\otimes \mathds 1_{\tilde {F}}\otimes \mathds{1}_{\operatorname{aux}}^F\right)\left(I_A\otimes I_F\right)\ket{\psi_{AF}}
\end{align*}

Therefore
\begin{align*}
\lVert J_A\left(M-I_A^\dagger \tilde{M} \otimes \mathds{1}_{\operatorname{aux}}^A I_A \right)\otimes \mathds{1}_F|\ket{\psi_{AF}}\rVert_2
&=
\lVert J_A(M\otimes \mathds 1_F)\ket{\psi_{AF}} - I_AI_A^\dagger \left(\tilde{M} \otimes \mathds{1}_{\tilde F}\right) J_A\ket{\psi_{AF}}\rVert_2\\
&\leq
\underbrace{\lVert J_A(M\otimes \mathds 1_F)\ket{\psi_{AF}}- \left(\tilde{M} \otimes \mathds 1_{\tilde{F}}\right) \ket{\varphi_{\tilde A\tilde F}} \otimes \ket{\operatorname{aux}_A}\rVert_2}_{\le\varepsilon_A\ \text{by \autoref{def:robust_self_test}}}\\
&\quad+
 \lVert\left(\tilde{M} \otimes \mathds 1_{\tilde{F}}\right) \ket{\varphi_{\tilde A\tilde F}} \otimes \ket{\operatorname{aux}_A} - I_AI_A^\dagger \left(\tilde{M} \otimes \mathds{1}_{\tilde F}\right) J_A\ket{\psi_{AF}}\rVert_2.
\end{align*}

For the remaining term, insert $I_A I_A^\dagger \left(\tilde{M}\otimes \mathds 1_{\tilde{F}}\right)\ket{\varphi_{\tilde A\tilde F}} \otimes \ket{\operatorname{aux}_A} $ and use the triangle inequality
\begin{align*}
\lVert\left(\tilde{M} \otimes \mathds 1_{\tilde{F}}\right) \ket{\varphi_{\tilde A\tilde F}} \otimes &\ket{\operatorname{aux}_A} - I_AI_A^\dagger \left(\tilde{M} \otimes \mathds{1}_{\tilde F}\right) J_A\ket{\psi_{AF}}\rVert_2\\
&\qquad\leq
\lVert \left(\tilde{M} \otimes \mathds 1_{\tilde{F}}\right) \ket{\varphi_{\tilde A\tilde F}} \otimes \ket{\operatorname{aux}_A} - I_A I_A^\dagger \left(\tilde{M}\otimes \mathds 1_{\tilde{F}}\right)\ket{\varphi_{\tilde A\tilde F}} \otimes \ket{\operatorname{aux}_A}\rVert_2 \\
&\qquad \quad+
\lVert I_A I_A^\dagger \left(\tilde{M}\otimes \mathds 1_{\tilde{F}}\right)\ket{\varphi_{\tilde A\tilde F}} \otimes \ket{\operatorname{aux}_A}-  I_AI_A^\dagger \left(\tilde{M} \otimes \mathds{1}_{\tilde F}\right) J_A\ket{\psi_{AF}}\rVert_2.
\end{align*}

Since $\tilde{M}$ is a contraction and $I_AI_A^\dagger$ is a projector, the second term is bounded by \autoref{def:robust_self_test}.
\begin{align}
    \lVert I_A I_A^\dagger \left(\tilde{M}\otimes \mathds 1_{\tilde{F}}\right)\ket{\varphi_{\tilde A\tilde F}} \otimes \ket{\operatorname{aux}_A}-  I_AI_A^\dagger \left(\tilde{M} \otimes \mathds{1}_{\tilde F}\right) J_A\ket{\psi_{AF}}\rVert_2\leq \varepsilon_A
\end{align}

We are left with finding a bound on 
\begin{equation}
\begin{aligned}
    \lVert &\left(\tilde{M} \otimes \mathds 1_{\tilde{F}}\right) \ket{\varphi_{\tilde A\tilde F}} \otimes \ket{\operatorname{aux}_A} - I_A I_A^\dagger \left(\tilde{M}\otimes \mathds 1_{\tilde{F}}\right)\ket{\varphi_{\tilde A\tilde F}} \otimes \ket{\operatorname{aux}_A}\rVert_2 \\
    &=\lVert \left(1-I_AI_A^\dagger\right) \left(\tilde{M} \otimes \mathds 1_{\tilde{F}}\right) \ket{\varphi_{\tilde A\tilde F}} \otimes \ket{\operatorname{aux}_A} \rVert_2
\end{aligned}
\end{equation}
Now we use that 
\begin{align}
    \left( 1 - I_AI_A^\dagger\right) J_A \left(M\otimes \mathds{1}_F\right)\psi_{AF} = 0,
\end{align}
because $\left( 1 - I_AI_A^\dagger\right) I_A \otimes I_F = 0$. Thus, we have
\begin{align}
    \lVert \left(1-I_AI_A^\dagger\right) &\left(\tilde{M} \otimes \mathds 1_{\tilde{F}}\right) \ket{\varphi_{\tilde A\tilde F}} \otimes \ket{\operatorname{aux}_A}  \rVert_2 \\
    &= \lVert \left(1-I_AI_A^\dagger\right) \left(\left(\tilde{M} \otimes \mathds 1_{\tilde{F}}\right) \ket{\varphi_{\tilde A\tilde F}} \otimes \ket{\operatorname{aux}_A} - J_A \left(M\otimes \mathds{1}_F\right)\psi_{AF}\right) \rVert_2 \\
    &\leq \lVert \left(\tilde{M} \otimes \mathds 1_{\tilde{F}}\right) \ket{\varphi_{\tilde A\tilde F}} \otimes \ket{\operatorname{aux}_A} - J_A \left(M\otimes \mathds{1}_F\right)\psi_{AF} \rVert_2 \\
    &\leq \varepsilon_{A}.
\end{align}

Now we consider $\rho_{AB}$ be the given state between Alice and Bob with the marginal constraint $\rho_A=\psi_A$ and let $\ket{\Psi_{ABR}}$ be an arbitrary purification of $\rho_{AB}$. Since $\ket{\psi_{AF}}$ is also a purification of $\rho_A$, Uhlmann's theorem implies the existence of an isometry
\begin{align}
V:\mathcal{H}_F\to \mathcal{H}_B\otimes \mathcal{H}_R
\quad\text{such that}\quad
\ket{\Psi}_{ABR} = \left(\mathds{1}_A\otimes V\right)\ket{\psi_{AF}}.
\end{align}
Consider the first term in \autoref{eq:split_error}
\begin{align}
\rvert\tr\left[\rho_{AB}\left(M-\Lambda_A^*\left(\tilde{M}\right)\right)\otimes N\right]\rvert
=
\bra{\Psi}\left(M-\Lambda_A^*\left(\tilde{M}\right)\right)\otimes N \otimes \mathds 1_R\ket{\Psi}\rvert.
\end{align}
Now we apply Uhlmann's theorem and define $X \coloneqq   V^\dagger N \otimes \mathds{1}_R V \leq 1$ with Cauchy-Schwarz inequality
\begin{equation}
    \begin{aligned}
        \bra{\Psi}\left(M-\Lambda_A^*\left(\tilde{M}\right)\right)\otimes N \otimes \mathds 1_R\ket{\Psi}\rvert &= \lvert \bra{\psi_{AF}}\mathds{1}_A \otimes V^\dagger \left[\left(M-\Lambda_A^*\left(\tilde{M}\right)\right)\otimes N \otimes \mathds{1}_R\right] \mathds{1}_A \otimes V \ket{\psi_{AF}}\rvert \\
        &= \lvert \bra{\psi_{AF}} \left(M-\Lambda_A^*\left(\tilde{M}\right)\right)\otimes X \ket{\psi_{AF}}\rvert \\
        &\leq \lVert \left(M-\Lambda_A^*\left(\tilde{M}\right)\right)\otimes \mathds{1}_F \ket{\psi_{AF}}\rVert_2 \cdot \Vert \mathds{1}_A \otimes X \ket{\psi_{AF}} \Vert_2 \\
        &\leq  \lVert \left(M-\Lambda_A^*\left(\tilde{M}\right)\right)\otimes \mathds{1}_F \ket{\psi_{AF}}\rVert_2. 
    \end{aligned}
\end{equation}
Now we apply \autoref{eq:3-eps-bound-on-A}.
Repeating the argument of the previous part with the roles $(A,F,\psi,\varepsilon_A,M,\tilde M)$ replaced by $(B,G,\mu,\varepsilon_B,N,\tilde N)$ yields the state-dependent vector bound
\begin{align}
\lVert (N-\Lambda_B^*(\tilde N))\otimes \mathds 1_G\ \ket{\mu_{BG}}\rVert_2\ \le\ 3\varepsilon_B,
\end{align}
and, using the marginal constraint $\rho_B=\mu_B$ from \autoref{eq:marginal_constraint}, the correlation bound
\begin{align}\label{eq:second_term_bound_general}
\lvert \tr\left[\rho_{AB}\,X\otimes \left(N-\Lambda_B^*\left(\tilde N\right)\right)\right]\rvert \ \le\ 3\varepsilon_B
\qquad\text{for every operator }X\text{ on }\mathcal{H}_A\text{ with }\lVert X\rVert\le 1.
\end{align}

We apply this with $X=\Lambda_A^*(\tilde M)$. Then \autoref{eq:second_term_bound_general} gives
\begin{align}\label{eq:second_term_bound}
\lvert \tr\left[\rho_{AB}\,\Lambda_A^*(\tilde M)\otimes \left(N-\Lambda_B^*\left(\tilde N\right)\right)\right]\rvert \ \le\ 3\varepsilon_B.
\end{align}

Finally, combine \autoref{eq:split_error} and \autoref{eq:second_term_bound}:
\begin{align}
\lvert
\tr\left[\rho_{AB}\,M_{a|x}\otimes N_{b|y}\right]
-
\tr\left[\tilde\rho_{\tilde A\tilde B}\,\tilde M_{a|x}\otimes \tilde N_{b|y}\right]
\rvert
\leq 3 \left(\varepsilon_A + \varepsilon_B\right).
\end{align}
This is exactly \autoref{eq:probability_error}.
\end{proof}

The following lemma shows that, given a local $\varepsilon_A$-robust self-test in Alice's lab, the conditional entropy can be transferred from the fully device-independent setting to a device-dependent one, up to an error controlled by the continuity properties of the conditional entropy. In general, this is a very strong statement, and one should not expect it to be efficient outside a small neighborhood of the self-test. In essence, it shows that the conditional entropy remains well behaved even when viewed solely as a function of the underlying statistics. By contrast, within the convex set of quantum-realizable statistics, for example those arising in a quantum game, many different strategies can produce the same observed statistics. In such situations, it is therefore unlikely that the conditional entropy can, in general, be reduced to an effectively device-dependent problem.
\begin{lemma}[Entropy transfer]\label{lem:entropy_transfer_continuity}
Fix $\tilde x\in X$ and assume \autoref{assumptions}. Let $\ketbra{\psi}{\psi}_{ABE}$
be a purification of $\rho_{AB}$, and define
$E' \coloneqq \mathcal H^{A}_{\operatorname{aux}}\otimes \mathcal H^{B}_{\operatorname{aux}}\otimes \mathcal H_E.$
Define the implemented classical-quantum state for $\tilde x$ by
\begin{align}
\rho^{\operatorname{phys}}_{AE}
\coloneqq
\sum_{a\in A}\ketbra{a}{a}_{A}\otimes
\tr_A\!\left[(M_{a\vert \tilde x}\otimes \mathds{1}_E)\rho_{AE}\right].
\end{align}
Define also the classical-quantum state on $AE'$ with 
\begin{align}
\ket{\Psi}
\coloneqq
(I_A\otimes I_B\otimes \mathds{1}_E)\ket{\psi}_{ABE} \in \mathcal{H}_{\tilde{A}} \otimes \mathcal{H}^A_{\operatorname{aux}} \otimes \tilde{\mathcal{H}}_B \otimes \mathcal{H}^B_{\operatorname{aux}} \otimes \mathcal{H}_E
\end{align}
by
\begin{align}
\tau_{AE'}
\coloneqq
\sum_{a\in A}\ketbra{a}{a}_{A}\otimes
\tr_{\tilde A\tilde B}\!\left[
\Bigl(
(\tilde M_{a\vert \tilde x}\otimes \mathds{1}^{A}_{\operatorname{aux}})
\otimes \mathds{1}_{\tilde B\,\mathcal H^{B}_{\operatorname{aux}}\,E}
\Bigr)
\ketbra{\Psi}{\Psi}
\right].
\end{align}
Then, with $\delta\coloneqq |A|\varepsilon_A$,
\begin{align}
\mathbb{H}(A\vert E)_{\rho^{\operatorname{phys}}_{AE}}
\ge
\mathbb{H}(A\vert E')_{\tau_{AE'}}
-
f(\delta,|A|).
\end{align}
\end{lemma}
\begin{proof}
Since $\rho_A=\psi_A$ by \autoref{eq:MC}, and $\ket{\psi_{AF}}$ is a purification of $\rho_A$, while $\ket{\psi}_{ABE}$ is a purification of $\rho_A$ as well, Uhlmann's theorem yields an isometry
\begin{align}
W:\mathcal H_F \to \mathcal H_B\otimes \mathcal H_E
\qquad\text{such that}\qquad
\ket{\psi}_{ABE}=(\mathds{1}_A\otimes W)\ket{\psi_{AF}}.
\end{align}
Define, in addition to $\tau_{AE'}$,
\begin{align}\label{eq:marginal_constraint_used_entropy_transer}
\omega_{AE'}
\coloneqq
\sum_{a\in A}\ketbra{a}{a}_{A}\otimes
\tr_{\tilde A\tilde B}\!\left[
\Bigl(
(I_A M_{a\vert \tilde x} I_A^\dagger)\otimes
\mathds{1}_{\tilde B\,\mathcal H^B_{\operatorname{aux}}\,E}
\Bigr)
\ketbra{\Psi}{\Psi}
\right].
\end{align}
Thus $\omega_{AE'}$ is the lifted cq-state obtained from the physical measurement, whereas $\tau_{AE'}$ is obtained from the ideal measurement.

We first derive an intertwining estimate from the robust local dilation on Alice's side. For this, set
\begin{align}
J_A \coloneqq I_A\otimes \mathds{1}_F
\end{align}
and
\begin{align}
\tilde T_a
\coloneqq
(\tilde M_{a\vert \tilde x}\otimes \mathds{1}^{A}_{\operatorname{aux}})
\otimes
\mathds{1}_{\tilde F\,F_{\operatorname{aux}}}.
\end{align}
By $S_A\succ_{\varepsilon_A}\tilde S_A$, there exists an auxiliary state
$\ket{\operatorname{aux}} \coloneqq \ket{\operatorname{aux}}_A \otimes \ket{\operatorname{aux}}_F\in \mathcal H_{F_{\operatorname{aux}}}$ such that
\begin{align}
\left\|
J_A(M_{a\vert \tilde x}\otimes \mathds{1}_F)\ket{\psi_{AF}}
-
\tilde T_a\bigl(\ket{\varphi_{\tilde A\tilde F}}\otimes \ket{\operatorname{aux}}\bigr)
\right\|_2
&\le \varepsilon_A,
\\
\left\|
J_A\ket{\psi_{AF}}
-
\ket{\varphi_{\tilde A\tilde F}}\otimes \ket{\operatorname{aux}}
\right\|_2
&\le \varepsilon_A.
\end{align}
Since $\tilde T_a$ is a contraction, applying $\tilde T_a$ to the second estimate gives
\begin{align}
\left\|
\tilde T_a J_A\ket{\psi_{AF}}
-
\tilde T_a\bigl(\ket{\varphi_{\tilde A\tilde F}}\otimes \ket{\operatorname{aux}}\bigr)
\right\|_2
\le \varepsilon_A.
\end{align}
Hence, by the triangle inequality,
\begin{align}
\left\|
J_A(M_{a\vert \tilde x}\otimes \mathds{1}_F)\ket{\psi_{AF}}
-
\tilde T_a J_A\ket{\psi_{AF}}
\right\|_2
\le 2\varepsilon_A.
\end{align}
Using the definitions of $J_A$ and $\tilde T_a$, this is
\begin{align}
\left\|
\Bigl(
I_A M_{a\vert \tilde x}
-
(\tilde M_{a\vert \tilde x}\otimes \mathds{1}^{A}_{\operatorname{aux}})I_A
\Bigr)\otimes \mathds{1}_F
\,\ket{\psi_{AF}}
\right\|_2
\le 2\varepsilon_A.
\end{align}
Now apply the isometry $\mathds{1}_A\otimes W$ and use
$\ket{\psi}_{ABE}=(\mathds{1}_A\otimes W)\ket{\psi_{AF}}$ to obtain
\begin{align}\label{eq:entropy_transfer_intertwine_ABE}
\left\|
\Bigl(
I_A M_{a\vert \tilde x}
-
(\tilde M_{a\vert \tilde x}\otimes \mathds{1}^{A}_{\operatorname{aux}})I_A
\Bigr)\otimes \mathds{1}_{BE}
\,\ket{\psi}_{ABE}
\right\|_2
\le 2\varepsilon_A.
\end{align}
Finally, since $\ket{\Psi}=(I_A\otimes I_B\otimes \mathds{1}_E)\ket{\psi}_{ABE}$ and
$I_B$ is an isometry, \autoref{eq:entropy_transfer_intertwine_ABE} implies
\begin{align}\label{eq:entropy_transfer_intertwine_lifted}
\left\|
\Bigl(
I_A M_{a\vert \tilde x} I_A^\dagger
-
(\tilde M_{a\vert \tilde x}\otimes \mathds{1}^{A}_{\operatorname{aux}})
\Bigr)\otimes \mathds{1}_{\tilde B\,\mathcal H^B_{\operatorname{aux}}\,E}
\,\ket{\Psi}
\right\|_2
\le 2\varepsilon_A.
\end{align}
We next bound the trace distance between $\omega_{AE'}$ and $\tau_{AE'}$.
Write
\begin{align}
\omega_{AE'}-\tau_{AE'}
=
\sum_{a\in A}\ketbra{a}{a}_{A}\otimes \Delta_a,
\end{align}
where
\begin{align}
\Delta_a
\coloneqq
\tr_{\tilde A\tilde B}\!\left[
\Bigl(
(I_A M_{a\vert \tilde x} I_A^\dagger
-
(\tilde M_{a\vert \tilde x}\otimes \mathds{1}^{A}_{\operatorname{aux}}))
\otimes \mathds{1}_{\tilde B\,\mathcal H^B_{\operatorname{aux}}\,E}
\Bigr)
\ketbra{\Psi}{\Psi}
\right].
\end{align}
Since the classical blocks are orthogonal,
\begin{align}
\Vert \omega_{AE'}-\tau_{AE'}\Vert_1
=
\sum_{a\in A}\Vert \Delta_a\Vert_1.
\end{align}
For fixed $a$, define
\begin{align}
\ket{u_a}
\coloneqq
\Bigl(
(I_A M_{a\vert \tilde x} I_A^\dagger
-
(\tilde M_{a\vert \tilde x}\otimes \mathds{1}^{A}_{\operatorname{aux}}))
\otimes \mathds{1}_{\tilde B\,\mathcal H^B_{\operatorname{aux}}\,E}
\Bigr)\ket{\Psi}.
\end{align}
Then $\Delta_a=\tr_{\tilde A\tilde B}\ketbra{u_a}{\Psi}$. Using the standard estimate
\begin{align}
\left\|\tr_X\ketbra{u}{v}\right\|_1 \le \|u\|_2\,\|v\|_2,
\end{align}
we get
\begin{align}
\|\Delta_a\|_1
\le
\|u_a\|_2\,\|\Psi\|_2
=
\|u_a\|_2.
\end{align}
Hence, by \autoref{eq:entropy_transfer_intertwine_lifted},
\begin{align}
\|\Delta_a\|_1 \le 2\varepsilon_A.
\end{align}
Summing over $a\in A$ yields
\begin{align}\label{eq:entropy_transfer_trace_distance}
\Vert \omega_{AE'}-\tau_{AE'}\Vert_1
\le 2|A|\,\varepsilon_A,
\qquad
\frac12\Vert \omega_{AE'}-\tau_{AE'}\Vert_1
\le |A|\,\varepsilon_A
\eqqcolon \delta.
\end{align}

We now compare the relevant marginals.

First, by construction of $\omega_{AE'}$, applying the local isometries and then
measuring the lifted effects $I_A M_{a\vert \tilde x}I_A^\dagger$ produces the same
physical cq-state as measuring $M_{a\vert \tilde x}$ before the isometries.
Therefore,
\begin{align}\label{eq:entropy_transfer_physical_marginal}
\tr_{E'\setminus E}\!\left[\omega_{AE'}\right]
=
\rho^{\operatorname{phys}}_{AE}.
\end{align}

Second, $\omega_{AE'}$ and $\tau_{AE'}$ have the same conditioning marginal.
Indeed,
\begin{align}
\omega_{E'}
&=
\sum_{a\in A}
\tr_{\tilde A\tilde B}\!\left[
\Bigl(
(I_A M_{a\vert \tilde x} I_A^\dagger)\otimes
\mathds{1}_{\tilde B\,\mathcal H^B_{\operatorname{aux}}\,E}
\Bigr)\ketbra{\Psi}{\Psi}
\right]
\\
&=
\tr_{\tilde A\tilde B}\!\left[
\Bigl(
(I_A I_A^\dagger)\otimes
\mathds{1}_{\tilde B\,\mathcal H^B_{\operatorname{aux}}\,E}
\Bigr)\ketbra{\Psi}{\Psi}
\right]
\\
&=
\tr_{\tilde A\tilde B}\!\left[\ketbra{\Psi}{\Psi}\right],
\end{align}
because $\sum_a M_{a\vert \tilde x}=\mathds{1}_A$ and $\ket{\Psi}$ lies in the range
of $I_A\otimes I_B$. Similarly,
\begin{align}
\tau_{E'}
&=
\tr_{\tilde A\tilde B}\!\left[
\Bigl(
(\sum_a \tilde M_{a\vert \tilde x}\otimes \mathds{1}^{A}_{\operatorname{aux}})
\otimes
\mathds{1}_{\tilde B\,\mathcal H^B_{\operatorname{aux}}\,E}
\Bigr)\ketbra{\Psi}{\Psi}
\right]
\\
&=
\tr_{\tilde A\tilde B}\!\left[\ketbra{\Psi}{\Psi}\right].
\end{align}
Hence
\begin{align}\label{eq:entropy_transfer_equal_marginal}
\omega_{E'}=\tau_{E'}.
\end{align}
We can now invoke the two abstract properties of the conditional entropy $\mathbb H$.

By \autoref{eq:entropy_transfer_physical_marginal}, the state
$\rho^{\operatorname{phys}}_{AE}$ is obtained from $\omega_{AE'}$ by a local
processing of the conditioning register $E'$. Therefore, monotonicity under local
processing of the conditioning system gives
\begin{align}
\mathbb H(A\vert E)_{\rho^{\operatorname{phys}}_{AE}}
\ge
\mathbb H(A\vert E')_{\omega_{AE'}}.
\end{align}
Moreover, by \autoref{eq:entropy_transfer_equal_marginal} and
\autoref{eq:entropy_transfer_trace_distance}, the continuity bound at fixed
conditioning marginal yields
\begin{align}
\mathbb H(A\vert E')_{\omega_{AE'}}
\ge
\mathbb H(A\vert E')_{\tau_{AE'}}
-
f(\delta,|A|).
\end{align}
Combining the two inequalities proves
\begin{align}
\mathbb H(A\vert E)_{\rho^{\operatorname{phys}}_{AE}}
\ge
\mathbb H(A\vert E')_{\tau_{AE'}}
-
f(\delta,|A|),
\end{align}
with $\delta=|A|\,\varepsilon_A$.
\end{proof}

For the device-dependent model $\left(\mathcal{H}_{\tilde{A}},\tilde{\mathcal{H}}_B, \{\tilde{M}_{a\vert x}\}, \{\tilde{N}_{b\vert y}\}\right)$, we consider the following optimization problem on the reduced space coming from \autoref{assumptions}. 
\begin{equation}\label{eq:OPT_ideal_entropy_relaxed}
\begin{aligned}
\widetilde{\operatorname{OPT}}^{\mathbb{H}}_{\varepsilon_{\operatorname{prob}}}(p;\tilde x)
\coloneqq 
\inf\
& \mathbb H\left(A_{\tilde x}\vert E'\right)_{(\tilde{\mathcal M}_{\tilde x}\otimes \id_{E'})(\tilde\rho_{\tilde A E'})}\\
\operatorname{s.t.}\
& \left\vert
\tr \left[\tilde\rho_{\tilde A\tilde B}\,\tilde M_{a\vert x}\otimes \tilde N_{b\vert y}\right]
-
p_{abxy}
\right\vert
\leq
3(\varepsilon_A + \varepsilon_B),
\quad  a,b,x,y,\\
& \tilde\rho_{\tilde A\tilde B E'} \geq 0.
\end{aligned}
\end{equation}
Then we can show the following theorem. 
\begin{theorem}[Finite-dimensional reduction from robust self-testing]\label{thm:finite_reduction_from_RST}
Assume fixed finite-dimensional ideal projective families
$\{\tilde M_{a\vert x}\}$ on $\tilde{\mathcal H}_A$ and
$\{\tilde N_{b\vert y}\}$ on $\tilde{\mathcal H}_B$.

Assume that for every feasible realization of \autoref{eq:OPT-DI},
there exist local $\varepsilon_A$- and $\varepsilon_B$-dilations as in
Definition~\ref{def:robust_self_test}, and that
Lemma~\ref{lem:compression_correlations} applies, yielding
\begin{align}
\varepsilon_{\operatorname{prob}}=3\left(\varepsilon_A+\varepsilon_B\right).
\end{align}
Assume moreover the entropy-transfer estimate
\begin{align}
\mathbb H\left(A_{\tilde x}\vert E\right)_{\operatorname{phys}}
\ge
\mathbb H\left(A_{\tilde x}\vert E'\right)_{\operatorname{ideal}}
-
f(\varepsilon_A,\vert A\vert).
\end{align}
Then
\begin{align}\label{eq:main_reduction_entropy}
\operatorname{OPT}_{\mathrm{DI}}^{\mathbb H}(p;\tilde x)
\ge
\widetilde{\operatorname{OPT}}^{\mathbb H}_{\varepsilon_{\operatorname{prob}}}(p;\tilde x)
-
f(\varepsilon_A,\vert A\vert).
\end{align}
\end{theorem}

\begin{proof}
Fix an arbitrary feasible realization of \autoref{eq:OPT-DI},
\begin{align}
\left(\rho_{ABE},\{M_{a\vert x}\},\{N_{b\vert y}\}\right),
\qquad
\rho_{ABE}=\ketbra{\psi}{\psi}.
\end{align}
By assumption, there exist local $\varepsilon_A$- and $\varepsilon_B$-dilations with isometries $I_A,I_B$.
Define the compressed state as in \autoref{eq:RS},
\begin{align}
\tilde\rho_{\tilde A\tilde B}
=
\tr_{\operatorname{aux}}
\left[
\left(I_A\otimes I_B\right)\rho_{AB}\left(I_A\otimes I_B\right)^\dagger
\right].
\end{align}
Also define the lifted tripartite state
\begin{align}
\tilde\rho_{\tilde A\tilde B E'}
:=
\left(I_A\otimes I_B\otimes \mathds{1}_E\right)\rho_{ABE}
\left(I_A\otimes I_B\otimes \mathds{1}_E\right)^\dagger,
\end{align}
where the auxiliary output systems of $I_A$ and $I_B$ are absorbed into $E'$.
Its $\tilde A\tilde B$-marginal is precisely $\tilde\rho_{\tilde A\tilde B}$.

By Lemma~\ref{lem:compression_correlations}, for all $a,b,x,y$,
\begin{align}
\left\vert
\tr \left[\rho_{AB}M_{a\vert x}\otimes N_{b\vert y}\right]
-
\tr \left[\tilde\rho_{\tilde A\tilde B}\tilde M_{a\vert x}\otimes \tilde N_{b\vert y}\right]
\right\vert
\le
\varepsilon_{\operatorname{prob}}.
\end{align}
Since the original realization is feasible in \autoref{eq:OPT-DI},
\begin{align}
\tr \left[\rho_{AB}M_{a\vert x}\otimes N_{b\vert y}\right]=p_{abxy},
\end{align}
hence
\begin{align}
\left\vert
\tr \left[\tilde\rho_{\tilde A\tilde B}\tilde M_{a\vert x}\otimes \tilde N_{b\vert y}\right]
-
p_{abxy}
\right\vert
\le
\varepsilon_{\operatorname{prob}}
\quad \forall a,b,x,y.
\end{align}
Therefore the lifted ideal realization
$\bigl(\tilde\rho_{\tilde A\tilde B E'},\{\tilde M_{a\vert x}\},\{\tilde N_{b\vert y}\}\bigr)$
is feasible for \autoref{eq:OPT_ideal_entropy_relaxed}.

Now apply the assumed entropy-transfer bound to this realization: 
\begin{align}
\mathbb H\left(A_{\tilde x}\vert E\right)_{\operatorname{phys}}
\ge
\mathbb H\left(A_{\tilde x}\vert E'\right)_{\operatorname{ideal}}
-
f(\varepsilon_A,\vert A\vert).
\end{align}
The term
\begin{align}
\mathbb H\left(A_{\tilde x}\vert E'\right)_{\operatorname{ideal}}
\end{align}
is exactly the objective value of \autoref{eq:OPT_ideal_entropy_relaxed} at a feasible point, so
\begin{align}
\mathbb H\left(A_{\tilde x}\vert E'\right)_{\operatorname{ideal}}
\ge
\widetilde{\operatorname{OPT}}^{\mathbb H}_{\varepsilon_{\operatorname{prob}}}(p;\tilde x).
\end{align}
Consequently,
\begin{align}
\mathbb H\left(A_{\tilde x}\vert E\right)_{\operatorname{phys}}
\ge
\widetilde{\operatorname{OPT}}^{\mathbb H}_{\varepsilon_{\operatorname{prob}}}(p;\tilde x)
-
f(\varepsilon_A,\vert A\vert).
\end{align}

Since the chosen feasible realization of \autoref{eq:OPT-DI} was arbitrary, taking the infimum over all feasible realizations yields
\autoref{eq:main_reduction_entropy}.
\end{proof}

\subsection{Error bounds for robust self-tests} \label{sec:CHSH-bounds}

In \autoref{assumptions}, we started with $\varepsilon_A$- and $\varepsilon_B$-dilations. One might therefore rightfully ask how to obtain the parameters $\varepsilon_A$, $\varepsilon_B$ which played a crucial role in our analysis. In this subsection, we will make explicit how they can be obtained if we use the CHSH inequality as our Bell inequality. In \cite{operator_algebra_self_test2024}, it was proven that for an $\varepsilon$-optimal strategy (meaning that the value of the game is at most $\varepsilon$-far from the optimal value, or in other words that the observed violation of the CHSH inequality is at most $\varepsilon$ far from the maximal quantum violation), $\varepsilon_A$ and $\varepsilon_B$ can be chosen as $\mathcal O(\sqrt{\varepsilon})$.
We will sketch the arguments in Section 5.2 of \cite{zhao2024tsirelson} that yield a concrete bound on the constant in front of $\sqrt{\epsilon}$.

For the CHSH game, its associated game polynomial is given by
\begin{equation}
    \Phi_{\mathrm{CHSH}}:=\frac{1}{2} + \frac{1}{8}[a_0 \otimes b_0 + a_0 \otimes b_1 + a_1 \otimes b_0 - a_1 \otimes b_1]
\end{equation}
Here, $a_0$, $a_1$ are the observables Alice measures and $b_0$, $b_1$ the observables Bob measures. Since the optimal quantum value of the CHSH game is $\frac{1}{2}+\frac{1}{2\sqrt{2}}\approx 0.85$, we can find a sum-of-squares decomposition of 
\begin{equation}
   \frac{1}{2}+\frac{1}{2\sqrt{2}} - \Phi_{\mathrm{CHSH}}=\frac{1}{8\sqrt{2}}\left[\left(a_0 \otimes \mathds{1}-\mathds{1} \otimes \frac{b_0+b_1}{\sqrt{2}}\right)^2 + \left(a_1 \otimes \mathds{1}-\mathds{1}\otimes\frac{b_0-b_1}{\sqrt{2}}\right)^2\right] \, .
\end{equation}
That is, for an $\varepsilon$-optimal strategy, with observables $A_0$, $A_1$, $B_0$, $B_1$ and state $\ket{\psi}$, 
we have that 
\begin{align}
   \left\Vert A_0\otimes \mathds{1}\ket{\psi} - \mathds{1} \otimes \frac{B_0 + B_1}{\sqrt{2}}\ket{\psi}\right \Vert_2 &\leq \sqrt[4]{128} \sqrt{\varepsilon} \\
    \left\Vert A_1\otimes \mathds{1} \ket{\psi} - \mathds{1} \otimes \frac{B_0 - B_1}{\sqrt{2}}\ket{\psi}\right \Vert_2 &\leq \sqrt[4]{128} \sqrt{\varepsilon}
\end{align}
Thus, we obtain, applying the triangle inequality and using $A_0^2 = \mathds{1}=A_1^2$ and $\Vert\frac{B_0+B_1}{\sqrt{2}}\Vert \leq \sqrt{2}$,
\begin{align}
    \left\Vert \ket{\psi} - \left(\frac{B_0+B_1}{\sqrt{2}}\right)^2\ket{\psi}\right\Vert_2 &\leq (1+\sqrt{2})\sqrt[4]{128} \sqrt{\varepsilon}\\
    \left\Vert \ket{\psi} - \left(\frac{B_0-B_1}{\sqrt{2}}\right)^2\ket{\psi}\right\Vert_2 &\leq (1+\sqrt{2})\sqrt[4]{128} \sqrt{\varepsilon}
\end{align}
Noting that we can write $\mathds{1} \otimes (B_0B_1+B_1B_0) = \left(\frac{B_0+B_1}{\sqrt{2}}\right)^2 - \left(\frac{B_0-B_1}{\sqrt{2}}\right)^2$ 
we can estimate 
\begin{align}
    \Vert \mathds{1} \otimes (B_0B_1+B_1B_0) \ket{\psi}\Vert_2 \leq 2(1+\sqrt{2})\sqrt[4]{128}\sqrt{\varepsilon}.
\end{align}
Using Gowers-Hatami theorem, this implies that there exists an isometry $I_B$ such that
\begin{equation}
    \left \Vert \mathds{1} \otimes B_y \ket{\psi}- \mathds {1} \otimes I_B^\dagger \left( \tilde B_y\right) I_B \ket{\psi}\right \Vert_2 \leq 2(1+\sqrt{2})\sqrt[4]{128}\sqrt{\varepsilon},
\end{equation}
where $B_y = \frac{X+(-1)^y Z}{\sqrt{2}}$ for all $y \in \{0, 1\}$. Since the CHSH inequality is symmetric in Alice's and Bob's observables, it follows by a similar argument that there exists an isometry $I_A$ such that 
\begin{equation}
    \left \Vert A_x \otimes \mathds{1} \ket{\psi}- I_A^\dagger \left( \tilde A_x\right) I_A \ket{\psi}\right \Vert_2 \leq 2(1+\sqrt{2})\sqrt[4]{128}\sqrt{\varepsilon},
\end{equation}
where $A_x = X^{1-x}Z^x$ for all $x \in \{0, 1\}$. Decomposing the observables into POVMs, using $A_x = 2M_{0|x}-\mathds{1}$, $\tilde A_x = 2\tilde M_{0|x}-\mathds{1}$, etc., we obtain
\begin{align}
    \left \Vert M_{a|x} \otimes \mathds{1} \ket{\psi}- I_A^\dagger \left( \tilde M_{a|x}\right) I_A \ket{\psi}\right \Vert_2 &\leq (1+\sqrt{2})\sqrt[4]{128}\sqrt{\varepsilon}, \\
    \left \Vert \mathds{1} \otimes N_{b|y} \ket{\psi}- \mathds {1} \otimes I_B^\dagger \left( \tilde N_{b|y}\right) I_B \ket{\psi}\right \Vert_2 &\leq (1+\sqrt{2})\sqrt[4]{128}\sqrt{\varepsilon} 
\end{align}

Next, choosing representations $\pi_A$ and $\pi_B$ such that
\begin{align}
    &\pi_A(a_0)=X, &\pi_A(a_1)=Z \\
    &\pi_B(b_0)=\frac{X+Z}{\sqrt{2}}, &\pi_B(b_1)= \frac{X-Z}{\sqrt{2}}
\end{align}
and letting $\pi = \pi_A \otimes \pi_B$, we obtain
\begin{equation}
    \pi(\Phi_{\mathrm{CHSH}}) = \frac{1}{2} + \frac{\sqrt{2}}{8}(X\otimes X + Z \otimes Z).
\end{equation}
This matrix has eigenvalues $\frac{1}{2}+\frac{1}{2\sqrt{2}}$, $\frac{1}{2}$ and $\frac{1}{2}-\frac{1}{2\sqrt{2}}$, which leads to a spectral gap of $\frac{1}{2\sqrt{2}}$. Now, we are ready to use  \cite[Proposition 4.2.13]{zhao2024tsirelson} (also appearing as \cite[Proposition 4.11]{operator_algebra_self_test2024}) with $\delta = 4(1+\sqrt{2})\sqrt[4]{128}\sqrt{\varepsilon}$ and $\Delta = \frac{1}{2\sqrt{2}}$:
It follows that there are isometries $I_A:\mathcal{H}_A \to \mathcal{H}_{\tilde{A}} \otimes \mathcal{H}_A^{\operatorname{aux}}$ and $I_B:\mathcal{H}_B \to \tilde{\mathcal{H}}_B \otimes \mathcal{H}_B^{\operatorname{aux}}$ and a vector state $\ket{\operatorname{aux}} \in \mathcal{H}_A^{\operatorname{aux}} \otimes \mathcal{H}_B^{\operatorname{aux}}$, such that 
    \begin{equation}
        \begin{aligned}
            \Vert I_A \otimes I_B \left(M_{a\vert x}\otimes \mathds{1}_B \ket \psi\right) - \left(\tilde{M}_{a\vert x} \otimes \mathds{1}_B \ket{\tilde \psi}\right)\otimes \ket{\operatorname{aux}}\Vert_2 \leq 2 (1+\sqrt{2})\sqrt[4]{128}\sqrt{\varepsilon}+\frac{\sqrt{2}(\delta + \sqrt{\varepsilon})}{\Delta} \leq 111 \sqrt{\varepsilon} \\
            \Vert I_A \otimes I_B \left(\mathds{1}_A\otimes N_{b\vert y} \ket \psi\right) - \left(\mathds{1}_A \otimes \tilde{N}_{b\vert y}  \ket{\tilde \psi}\right)\otimes \ket{\operatorname{aux}}\Vert_2 \leq 2 (1+\sqrt{2})\sqrt[4]{128}\sqrt{\varepsilon}+\frac{\sqrt{2}(\delta + \sqrt{\varepsilon})}{\Delta} \leq 111 \sqrt{\varepsilon} \\
            \Vert I_A\otimes I_B \ket{\psi} - \ket{\tilde\psi} \otimes \ket{\operatorname{aux}}\Vert_2\leq \frac{\sqrt{2}(\delta + \sqrt{\varepsilon})}{\Delta}\leq 95 \sqrt{\varepsilon}
        \end{aligned}
\end{equation}
\section{A DIQKD protocol and finite-size analysis}
\subsection{The Protocol}
In this section we use our setup of local self-tests from \autoref{fig:two_switches_lift} in order to apply the results from \autoref{thm:finite_reduction_from_RST} for a sketch of a finite-size security argument. For this purpose we adopt the notion of conditional entropy from \cite[Def. 1]{hahn2025analyticrenyientropybounds}. Comparing \autoref{thm:finite_reduction_from_RST} to \cite[Eq.~3]{hahn2025analyticrenyientropybounds} yields that we can directly adopt the techniques from \cite[Sec.~III]{hahn2025analyticrenyientropybounds} if we can argue that the additional measurements we naturally have can be added to the protocol \cite[cryptographic protocol]{hahn2025analyticrenyientropybounds}, which we repeat in \autoref{prot:four_parties_two_switches}. As a protocol we consider the following. Alice holds $X = \{0,1\}$, Fred holds $\tilde{X}= \{0,1\}$, Bob holds $Y^\prime = \{0,1,2\}$ and George holds $\tilde{Y}= \{0,1\}.$ Each measurement has two outcomes. The important caveat is that we use for the lift just $Y = \{0,1\}$ such that the entropy optimization program from \autoref{thm:finite_reduction_from_RST} is a two-input, two-output program for Alice and Bob. Concretely we consider the following honest implementation on qubits in the $X$-$Z$-plane. We have the following measurements:
\begin{protocol}\raggedleft
\noindent\fbox{
\parbox{0.93\linewidth}{
\textbf{1. Measurements}: 
This step is carried out in $N$ rounds. The inputs to the rounds are drawn randomly and independently.\\
Alice produces a common random bit $T_i$ such that $P(T_i = 0) = 1- \gamma$ and $P(T_i= 1) = \gamma$. 

$T_i$ = 0:
if $T_i=0$ Alice and Bob choose their generation inputs, which are w.l.o.g. $X_i = 0$ and $Y_i = 2$. Fred and George execute a random measurement. 

$T_i$ = 1:
For Alice: $P(x=0)=p_A,\;P(x=1)=1-p_A$\\
For Bob: $P(y=0)=p_B,\;P(y=1)=1-p_B$\\
For Fred: $P(\tilde{x}=0)=p_F,\;P(\tilde{x}=1)=1-p_F$\\
For George: $P(\tilde{y}=0)=p_G,\;P(\tilde{y}=1)=1-p_G$\\
For Switch A: $P(T_A=0)=t_A,\;P(T_A=1)=1-t_A$\\
For Switch B: $P(T_B=0)=t_B,\;P(T_B=1)=1-t_B$\\[0.2cm]
\textbf{2. Sifting:} All inputs are announced over authenticated public channels. The binary random variable $T$ is used to select key-generation rounds. From rounds with $x=y=T_A=T_B=R=0$, Alice and Bob select a subset of length $\sim \gamma t_A t_B p_A p_B N$. 
The outputs $(a,b)$ of these rounds form a raw key. The data of all other rounds is kept for parameter estimation.\\[0.2cm]
\textbf{3. Parameter estimation:} The data from the test rounds is used to estimate the distribution $p^s(a,\tilde a,b,\tilde b\,|\,x,\tilde x,y,\tilde y)$ by frequencies. If this estimate deviates from a predefined expected behavior in $S_{\Omega}$, the protocol is aborted.\\[0.2cm]
\textbf{4. \& 5. Error correction and Privacy amplification}: 
One-way error correction from Alice to Bob and privacy amplification (cf., e.g. \cite{Tomamichel2017}.)
}
}
\caption{A spot-checking protocol involving four parties is shown. Each party has binary inputs and outputs, and the goal is to self-test both local devices of Alice and Bob in order to generate key material from BB84 measurements.}
\label{prot:four_parties_two_switches}
\end{protocol}

\begin{table}[ht]
    \centering
    \[
    \begin{array}{c|cccc}
         & \text{Alice} & \text{Bob} & \text{Fred} & \text{George} \\
        \hline
        0 & X & \tfrac{1}{\sqrt{2}}(X+Z) & \tfrac{1}{\sqrt{2}}(X+Z) & X \\
        1 & Z & \tfrac{1}{\sqrt{2}}(X-Z) & \tfrac{1}{\sqrt{2}}(X-Z) & Z
    \end{array}
    \]
    \caption{Measurement settings for Alice, Bob, Fred, and George.}
    \label{tab:measurement-settings}
\end{table}

These measurements are used for parameter estimation. 

As the physical measurements on Alice's and Bob's side are thus robust self-tests by e.g.~\cite[Thm.~9.5]{operator_algebra_self_test2024} and \autoref{sec:CHSH-bounds}, we can apply \autoref{thm:finite_reduction_from_RST}. Thus, assuming the results from the measurements in 
\begin{align}
    S_A \coloneqq    \left(\mathcal{H}_A, \mathcal{H}_F, \{M_{ a\vert x}\},\{P_{p\vert q}\},\ket{\psi_{AF}} \right) \quad \text{of} \quad  \tilde{S}_A \coloneqq    \left(\mathcal{H}_{\tilde{A}}, \mathcal{H}_{\tilde{F}}, \{\tilde{M}_{ a\vert x}\},\{\tilde{P}_{p\vert q}\},\ket{\varphi_{\tilde{A}\tilde{F}}} \right)
\end{align}
that is $S_A \succ_{\varepsilon_A} \tilde{S}_A$ and similarly for a subset $Y\subseteq Y^\prime$
\begin{align}
    S_B \coloneqq    \left(\mathcal{H}_B, \mathcal{H}_G, \{N_{b\vert y}\},\{Q_{r\vert s}\},\ket{\mu_{BG}} \right) \quad \text{of} \quad  \tilde{S}_B \coloneqq    \left(\mathcal{H}_{\tilde{B}}, \mathcal{H}_{\tilde{G}},\{\tilde{N}_{b\vert y}\},\{\tilde{Q}_{r\vert s}\} ,\ket{\nu_{\tilde{A}\tilde{F}}} \right)
\end{align}
are $\varepsilon_A$ and $\varepsilon_B$ robust self-tests, we have that the extractable randomness can be calculated from \autoref{eq:OPT_ideal_entropy_relaxed}, repeated in the following
\begin{equation}\label{eq:optimization_lift_repeat}
\begin{aligned}
\inf\
& \mathbb H\left(A_{\tilde x}\vert E'\right)_{(\tilde{\mathcal M}_{\tilde x}\otimes \id_{E'})(\tilde\rho_{\tilde A E'})}\\
\operatorname{s.t.}\
& \left\vert
\tr \left[\tilde\rho_{\tilde A\tilde B}\,\tilde M_{a\vert x}\otimes \tilde N_{b\vert y}\right]
-
p_{abxy}
\right\vert
\leq
3(\varepsilon_A + \varepsilon_B),
\quad  a,b,x,y,\\
& \tilde\rho_{\tilde A\tilde B E'} \geq 0.
\end{aligned}
\end{equation}
\subsection{Theoretical application of Renyi-entropy accumulation and secrecy}

After the execution of the protocol the outcome state is $\rho_{A_1^NB_1^N \tilde{A}_1^N \tilde{B}_1^N X_1^N {Y^\prime}_1^N \tilde{X}_1^N \tilde{Y}_1^N T_1^N E}$. Using the Renyi-entropy accumulation theorem from \cite{Arqand2025}, we have similar to \cite[Eq.~12-13]{hahn2025analyticrenyientropybounds} 
    \begin{align}\label{Eq: GREATBound}   
    \widetilde{H}_{\alpha}^{\uparrow}\left({A}_{1}^{N}\bar{{C}}_{1}^{N}| {X}_{1}^{N}{Y}_{1}^{N}{T}_{1}^{N}\mathbf{E}\right)_{\rho_{\vert \Omega_{AT}}} \geq N h_{\alpha} 
- \frac{\alpha}{\alpha-1} \log\frac{1}{\Pr[\Omega_{AT}]},
    \end{align}
   where $h_{\alpha}$ is a quantity satisfying
\begin{align}\label{eq:singlerndopt}
     h_{\alpha} \geq \inf_{\Lambda} \inf_{q\in S_\mathrm{acc}}   \frac{1}{\alpha-1}D\left(q \Vert p\right) + q(\perp)\widetilde{H}_{\alpha}^{\downarrow}\left(A|X=0,E\right).
\end{align}
Here $\bar{C}_1^N$ denotes Alice's test round data  and $\Lambda$ denotes similar to \cite{hahn2025analyticrenyientropybounds} all single-round strategies. 
Importantly, the event $\Omega_{AT}$ denotes the whole distribution set of all four parties, even though the objective itself does not rely on Fred and George. However, the constraints in this optimization are affected by Fred and George and thus we must taken them into account. Given now observed statistics, the optimization in \autoref{eq:singlerndopt} can be rewritten with \autoref{eq:optimization_lift_repeat} into a  device-dependent optimization using \autoref{thm:finite_reduction_from_RST} up to a $\varepsilon_A$-error and additional $\varepsilon_A+\varepsilon_B$-constraints.  

\section{Outlook}
In this work, we developed and investigated a  framework for the description of device independent QKD in a routed Bell-test setting (cf.~\autoref{subsec:srq}). In the first part of this work we compared the models described in earlier proposals \cite{Lim2013,Le_Roy_Deloison_2025,Tan_2024}, for three parties with the model in \cite{kossmann2025routedbelltestsarbitrarily} and showed in particular that the short-range-quantum correlations are equivalent to special cases of this model. Moreover, we gave a detailed description of the assumptions on the switch, including in which sense it is device-independent and that the marginal constraint introduced in \cite{kossmann2025routedbelltestsarbitrarily} is on one hand unavoidable and on the other robust in the sense that a small deviation in trace-distance does not lead to a general lack of security (cf.~\autoref{lem:cstar_continuity_state}). 

The analysis in the second part of this work makes precise how the intuition of local self-tests can be used to lift a routed device-independent description to an effectively device-dependent one, provided the relevant marginal constraints are satisfied (cf.~\autoref{lem:compression_correlations} and \autoref{lem:entropy_transfer_continuity}). In the ideal limit, this recovers the familiar device-dependent picture and the corresponding optimal BB84 behavior. For the non-ideal case, we show that this lift is indeed robust. This means that we can recover an effective device-dependent picture with an error term expressing the performance of the local self-tests (cf.~\autoref{thm:finite_reduction_from_RST}). This also extends in particular to situations in which the marginal constraints are only fulfilled approximately (cf. \autoref{eq:marginal_constraint_used_entropy_transer} and \autoref{eq:second_term_bound_general}). 
Taken together, this makes this ansatz a practical tool for the analysis of  experiments. We showed exemplary how this effective device-dependent description can be used in order to lift a state of the art finite size analysis to device independent settings. However, this is a direction we have only begun to explore. 

A central next step will be the application of our tools to predict the performance of realistic future experiments.
Here we especially have to compute and optimize the error thresholds on the local self-tests for a successful finite size key generation. An important step for this is to improve techniques for estimating the errors $\epsilon_A$ and $\epsilon_B$ (cf.~\cite{Bancal_2015,Kaniewski_2016}). So far, theses estimates are based on the score of a non-local game. We however expect a much better performance when this estimate is done based on the full measurement statistics. 

If these remaining ingredients can be brought under quantitative control, QKD based on routed Bell-tests  may offer a realistic route towards a next generation of quantum key distribution technology.

\section{Acknowledgments}
GK thanks Thomas Hahn for valuable discussions about \cite{hahn2025analyticrenyientropybounds}. GK thanks AB for hosting him at the Grenoble Computer Science Laboratory (LIG), where parts of this project were carried out. GK acknowledges support by the European Research Council (ERC Grant Agreement No. 948139). AB was supported by the ANR project PraQPV, grant number ANR-24-CE47-3023. Part of this material is based upon work supported by the Swedish Research Council under grant no. 2021-06594 while AB was in participating in the program ``Operator Algebras and Quantum Information" at Institut Mittag-Leffler in Djursholm, Sweden. RS is supported by the DFG under Germany's Excellence Strategy - EXC-2123 QuantumFrontiers-2 - 390837967 and SFB 1227 (DQ-mat), the Quantum Valley Lower Saxony, and the BMBF projects ATIQ, SEQUIN, Quics and CBQD.

The authors would like to thank Yuming Zhao and Simon Schmidt for explaining them how to obtain concrete bounds for robust self-tests as in \autoref{sec:CHSH-bounds}.

\bibliography{main}
\bibliographystyle{ultimate}

\appendix

\section{A qubit reduction}\label{appendix:qubit_reduction}
For the self-testing setup, the most important example of \autoref{def:effective_overlap} can be given as follows. 
\begin{example}\label{ex:effective_overlap_feasible}
Let $I:\mathcal{H}_A\to\mathcal{H}_{A'}$ be an isometry and fix distinguished outcomes
$a_0$ and $b_0$. Set $P\coloneqq   II^\dagger$ and define POVMs on $A'$ by
\begin{align}
\Pi_{a\mid Z} &\coloneqq    IM_{a\mid Z}I^\dagger \quad (a\neq a_0), &
\Pi_{a_0\mid Z} &\coloneqq    IM_{a_0\mid Z}I^\dagger + (\mathds 1_{A'}-P),\\
\Pi_{b\mid X} &\coloneqq    IM_{b\mid X}I^\dagger \quad (b\neq b_0), &
\Pi_{b_0\mid X} &\coloneqq    IM_{b_0\mid X}I^\dagger + (\mathds 1_{A'}-P).
\end{align}
Then $\sum_a\Pi_{a\mid Z}=\sum_b\Pi_{b\mid X}=\mathds 1_{A'}$ and
$I^\dagger \Pi_{a\mid Z} I=M_{a\mid Z}$, $I^\dagger \Pi_{b\mid X} I=M_{b\mid X}$
(using $I^\dagger(\mathds 1_{A'}-P)I=0$). With $U=I$ and $K'=\{\mathds 1_{A'}\}$,
this is a valid tuple in \autoref{def:effective_overlap}.
\end{example}

\begin{proposition}\label{prop:xxx}
    Let $(\mathcal{H}_A,\mathcal{H}_F,\{M_{a\vert x}\},\{N_{b\vert y}\},\psi_{AF})$ be a strategy with two inputs and outputs and define Alice's reflections
    \begin{align}
        X \coloneqq   M_{0\vert X}- M_{1\vert X}, \quad Z \coloneqq   M_{0\vert Z}- M_{1\vert Z}.
    \end{align}
    Let $I_A:\mathcal{H}_A \to \mathcal{H}_{A^\prime}$ be any isometry and set for the dilated POVM from \autoref{ex:effective_overlap_feasible} 
    \begin{align}
        X^\prime \coloneqq   \Pi_{0\vert X} - \Pi_{1\vert X}, \quad Z^\prime \coloneqq     \Pi_{0\vert Z} - \Pi_{1\vert Z} \quad \text{and} \quad \sigma_A^{\prime} \coloneqq   I_A \sigma_A I_A^\dagger,
    \end{align}
    with $\sigma_A = \tr_{F}[\ketbra{\psi_{AF}}{\psi_{AF}}]$ then 
    \begin{align}\label{eq:result_proposition}
        c^\star(\sigma_A,X,Z) \leq \frac{1}{2} + \frac{1}{4}\left(\tr[\sigma_A^\prime\{X^\prime,Z^\prime\}^2]\right)^{1/2}.
    \end{align}
\end{proposition}
\begin{proof}
W.l.o.g.\ we treat the $Z$-conditioned operators in the following, as the $X$-conditioned ones
follow analogously. Define
\begin{align}
\Pi_{a\vert Z} \coloneqq   I_A M_{a\vert Z} I_A^\dagger \quad (a\neq a_0),
\qquad
\Pi_{a_0\vert Z} \coloneqq   I_A M_{a_0\vert Z} I_A^\dagger + (\mathds 1_{A^\prime}-I_AI_A^\dagger).
\end{align}
Then $\Pi_{a\vert Z}\geq 0$ and
\begin{align}
\sum_a \Pi_{a\vert Z}=\mathds 1_{A^\prime}
\end{align}
(as in \autoref{ex:effective_overlap_feasible}).
If the operators $M_{a\vert Z}$ are projections, then for $a\neq a_0$ we have
\begin{align}
(I_A M_{a\vert Z}I_A^\dagger)^2
= I_A M_{a\vert Z}I_A^\dagger I_A M_{a\vert Z}I_A^\dagger
= I_A M_{a\vert Z}(I_A^\dagger I_A) M_{a\vert Z}I_A^\dagger
= I_A M_{a\vert Z}^2 I_A^\dagger
= I_A M_{a\vert Z}I_A^\dagger,
\end{align}
using $I_A^\dagger I_A=\mathds 1_A$.
Moreover, since $(\mathds 1_{A^\prime}-I_AI_A^\dagger)I_A=0$, the cross terms vanish and
$\Pi_{a_0\vert Z}$ is a projection as well:
\begin{align}
\Pi_{a_0\vert Z}^2
&=\left(I_AM_{a_0\vert Z}I_A^\dagger + (\mathds 1-I_AI_A^\dagger)\right)^2 \nonumber\\
&=(I_AM_{a_0\vert Z}I_A^\dagger)^2 + (\mathds 1-I_AI_A^\dagger)^2
+ I_AM_{a_0\vert Z}I_A^\dagger(\mathds 1-I_AI_A^\dagger)
+ (\mathds 1-I_AI_A^\dagger)I_AM_{a_0\vert Z}I_A^\dagger \nonumber\\
&= I_AM_{a_0\vert Z}I_A^\dagger + (\mathds 1-I_AI_A^\dagger)
=\Pi_{a_0\vert Z}.
\end{align}
(Here we used $(\mathds 1-I_AI_A^\dagger)^2=(\mathds 1-I_AI_A^\dagger)$ and
$I_AI_A^\dagger(\mathds 1-I_AI_A^\dagger)=0=(\mathds 1-I_AI_A^\dagger)I_AI_A^\dagger$.)
Thus, if $\{M_{a\vert Z}\}_a$ is a PVM, then so is $\{\Pi_{a\vert Z}\}_a$; in particular,
orthogonality holds. Furthermore, the constraint \autoref{eq:definition_overlap_2} is satisfied:
\begin{align}
I_A^\dagger \Pi_{a\vert Z} I_A
=
\begin{cases}
I_A^\dagger I_A M_{a\vert Z} I_A^\dagger I_A = M_{a\vert Z}, & a\neq a_0,\\
I_A^\dagger I_A M_{a_0\vert Z} I_A^\dagger I_A + I_A^\dagger(\mathds 1-I_AI_A^\dagger)I_A
= M_{a_0\vert Z}, & a=a_0,
\end{cases}
\end{align}
since $I_A^\dagger(\mathds 1-I_AI_A^\dagger)I_A=0$.
Hence the tuple $(U,X',Z',K')=(I_A,\{\Pi_{b\vert X}\}_b,\{\Pi_{a\vert Z}\}_a,\{\mathds 1_{A'}\})$
is feasible for the definition of $c^\star(\sigma_A,X,Z)$, which implies
$c^\star(\sigma_A,X,Z)\le c^\star(\sigma_A',X',Z')$.

We now prove the second inequality in \autoref{eq:result_proposition}.
Assume from now on that Alice's measurements are binary PVMs, so that
\begin{align}
p \coloneqq   \Pi_{0\vert Z},\qquad q \coloneqq   \Pi_{0\vert X}
\end{align}
are orthogonal projections and $\Pi_{1\vert Z}=\mathds 1-p$, $\Pi_{1\vert X}=\mathds 1-q$.
Let $\mathcal{A}\coloneqq   C^\star(p,q)\subseteq \mathcal{B}(\mathcal{H}_{A^\prime})$.
Viewing the inclusion $\mathcal{A}\hookrightarrow \mathcal{B}(\mathcal{H}_{A^\prime})$ as a
representation of $\mathcal{A}$, \cite[Lem.~1.8]{Raeburn_algebra_two_projections}
yields an orthogonal decomposition into reducing subspaces
\begin{align}\label{eq:RS_decomposition}
\mathcal{H}_{A^\prime}=\bigoplus_k \mathcal{H}_k,
\qquad \dim(\mathcal{H}_k)\le 2.
\end{align}
Because the sum in \autoref{eq:RS_decomposition} is an \emph{orthogonal direct sum}, the projections
$\{P_{A'}^k\}_k$ are pairwise orthogonal, satisfy $\sum_k P_{A'}^k=\mathds 1_{A'}$ and project onto $\mathcal{H}_k$. Moreover, we have $p\mathcal{H}_k\subseteq \mathcal{H}_k$ and $p(\mathcal{H}_k^\perp)\subseteq \mathcal{H}_k^\perp$,
and similarly for $q$. Equivalently, the orthogonal projection $P_{A'}^k$ onto $\mathcal{H}_k$
commutes with every operator that leaves $\mathcal{H}_k$ and $\mathcal{H}_k^\perp$ invariant.
In particular,
\begin{align}\label{eq:commuting}
P_{A'}^k p = pP_{A'}^k\qquad\text{and}\qquad P_{A'}^k q = qP_{A'}^k.
\end{align}

Define the restricted projections on $\mathcal{H}_k$ by
\begin{align}
p_k \coloneqq   P_{A'}^k p P_{A'}^k,\qquad q_k \coloneqq   P_{A'}^k q P_{A'}^k.
\end{align}
Because of \autoref{eq:commuting}, these are again orthogonal projections on $\mathcal{H}_k$:
\begin{align}
p_k^2 = P^k p P^k p P^k = P^k p^2 P^k = P^k p P^k = p_k,
\qquad p_k^\dagger=p_k,
\end{align}
and analogously for $q_k$.

We now choose the projective measurement $K'=\{P_{A'}^k\}_k$ in \autoref{def:effective_overlap}
together with the trivial isometry $U=1_{A'}$ and the same POVMs $X'= \{\Pi_{b\vert X}\}_b$,
$Z'=\{\Pi_{a\vert Z}\}_a$. This is feasible since, using \autoref{eq:commuting} and $\sum_kP_{A'}^k=\mathds 1$,
\begin{align}
\sum_k U^\dagger P_{A'}^k \Pi_{a\vert Z} P_{A'}^k U
= \sum_k P_{A'}^k \Pi_{a\vert Z} P_{A'}^k
= \sum_k P_{A'}^k \Pi_{a\vert Z}
=\left(\sum_k P_{A'}^k\right)\Pi_{a\vert Z}
=\Pi_{a\vert Z},
\end{align}
and analogously for $\Pi_{b\vert X}$.
Inserting this feasible tuple into \autoref{def:effective_overlap} yields
\begin{align}\label{eq:definition_overlap_blocksum}
c^\star(\sigma_A',X',Z')
\le
\sum_k \tr\bigl[P_{A'}^{k}\,\sigma_A'\bigr]\;
\max_{b\in\{0,1\}}\lVert
\sum_{a\in\{0,1\}}
\bigl(P_{A'}^{k}\Pi_{a\mid Z}P_{A'}^{k}\bigr)
\bigl(P_{A'}^{k}\Pi_{b\mid X}P_{A'}^{k}\bigr)
\bigl(P_{A'}^{k}\Pi_{a\mid Z}P_{A'}^{k}\bigr)
\rVert.
\end{align}

If $\dim\mathcal{H}_k=1$, the inner term is trivial (the restricted projections are $0$ or $1$).
So assume $\dim\mathcal{H}_k=2$.
By \cite[Thm.~1.3]{Raeburn_algebra_two_projections}, there exists an orthonormal basis of
$\mathcal{H}_k\simeq \mathbb{C}^2$ such that for some $x_k\in(0,1)$
\begin{align}\label{eq:RS_matrices}
p_k=\begin{pmatrix}1&0\\0&0\end{pmatrix},
\qquad
q_k=\begin{pmatrix}
x_k & \sqrt{x_k(1-x_k)}\\
\sqrt{x_k(1-x_k)} & 1-x_k
\end{pmatrix}.
\end{align}
Moreover, by \cite[Rem.~1.4]{Raeburn_algebra_two_projections} we may write $x_k=\cos^2\vartheta_k$
with $\vartheta_k\in(0,\pi/2)$, and then $q_k$ is the projection onto
$\mathrm{span}\{(\cos\vartheta_k,\sin\vartheta_k)\}$.
Define the Bloch-sphere angle $\theta_k\coloneqq   2\vartheta_k\in(0,\pi)$; then
\begin{align}
\cos\theta_k=\cos(2\vartheta_k)=2\cos^2\vartheta_k-1=2x_k-1.
\end{align}

We now compute the inner norm in \autoref{eq:definition_overlap_blocksum}. Let $\Pi_{a\vert Z,k}:=P_k\Pi_{a\vert Z}P_k$ be the action of the measurement operators on the $k$-block.
For $b=0$,
and using $\Pi_{0\vert Z,k}=p_k$, $\Pi_{1\vert Z,k}=\mathds 1-p_k$, $\Pi_{0\vert X,k}=q_k$,
we obtain the $Z$-pinching of $q_k$:
\begin{align}
\sum_{a\in\{0,1\}}
\bigl(P_k\Pi_{a\mid Z}P_k\bigr)\bigl(P_k\Pi_{0\mid X}P_k\bigr)\bigl(P_k\Pi_{a\mid Z}P_k\bigr)
&=
p_k q_k p_k + (\mathds 1-p_k)q_k(\mathds 1-p_k)\nonumber\\
&=
\begin{pmatrix}
x_k & 0\\
0 & 1-x_k
\end{pmatrix}.
\end{align}
Hence
\begin{align}\label{eq:block_norm}
\max_{b\in\{0,1\}}
\lVert
\sum_{a\in\{0,1\}}
\bigl(P_k\Pi_{a\mid Z}P_k\bigr)\bigl(P_k\Pi_{b\mid X}P_k\bigr)\bigl(P_k\Pi_{a\mid Z}P_k\bigr)
\rVert
=
\max\{x_k,1-x_k\}
=
\frac{1}{2}+\frac{1}{2}|\cos\theta_k|.
\end{align}
(For $b=1$ the same norm is obtained by replacing $q_k$ with $\mathds 1-q_k$, which swaps the diagonal entries.)

Set $w_k\coloneqq   \tr[\sigma_A'P_k]$ (so $w_k\geq 0$ and $\sum_k w_k=1$).
Combining \autoref{eq:definition_overlap_blocksum} and \autoref{eq:block_norm} gives
\begin{align}\label{eq:cstar_abs_cos}
c^\star(\sigma_A',X',Z') \le \sum_k w_k\left(\frac{1}{2}+\frac{1}{2}|\cos\theta_k|\right)
= \frac{1}{2}+\frac{1}{2}\sum_k w_k|\cos\theta_k|.
\end{align}

We now relate $|\cos\theta_k|$ to the anticommutator. Define the reflections
\begin{align}
X'_k\coloneqq   P_k X' P_k = (2q_k-\mathds 1),\qquad
Z'_k\coloneqq   P_k Z' P_k = (2p_k-\mathds 1)
\quad\text{on }\mathcal{H}_k.
\end{align}
On $\mathbb{C}^2$, there exist unit vectors $\vec a_k,\vec b_k\in\mathbb{R}^3$ such that
$X'_k=\vec a_k\cdot\vec\sigma$ and $Z'_k=\vec b_k\cdot\vec\sigma$.
Using $(\vec u\cdot\vec\sigma)(\vec v\cdot\vec\sigma)=(\vec u\cdot\vec v)1
+i(\vec u\times \vec v)\cdot\vec\sigma$, we get, see also \cite{Raeburn_algebra_two_projections},
\begin{align}
\{X'_k,Z'_k\}=X'_kZ'_k+Z'_kX'_k = 2(\vec a_k\cdot\vec b_k)\mathds 1
=2\cos\theta_k\,\mathds 1,
\end{align}
hence $\{X'_k,Z'_k\}^2 = 4\cos^2\theta_k\,\mathds 1$.

Since each $P_k$ reduces $X'$ and $Z'$, it also reduces $\{X',Z'\}$, and thus
\begin{align}
\{X',Z'\}^2 P_k = \{X'_k,Z'_k\}^2 = 4\cos^2\theta_k\,P_k.
\end{align}
Therefore
\begin{align}\label{eq:trace_block_identity}
\tr\left[\sigma_A'\{X',Z'\}^2\right]
&=\sum_k \tr\left[\sigma_A'\{X',Z'\}^2P_k\right]
= \sum_k \tr\left[\sigma_A'(4\cos^2\theta_k P_k)\right]
=4\sum_k w_k\cos^2\theta_k.
\end{align}

Finally, apply concavity of $\sqrt{\cdot}$ to \autoref{eq:cstar_abs_cos}:
\begin{align}
\sum_k w_k|\cos\theta_k|
\le \left(\sum_k w_k\cos^2\theta_k\right)^{1/2}
=\frac{1}{2}\left(\tr\left[\sigma_A'\{X',Z'\}^2\right]\right)^{1/2},
\end{align}
where the last step uses \autoref{eq:trace_block_identity}.
Substituting into \autoref{eq:cstar_abs_cos} yields
\begin{align}
c^\star(\sigma_A',X',Z')
\le
\frac{1}{2}+\frac{1}{2}\cdot\frac{1}{2}\left(\tr\left[\sigma_A'\{X',Z'\}^2\right]\right)^{1/2}
=
\frac{1}{2}+\frac14\left(\tr\left[\sigma_A'\{X',Z'\}^2\right]\right)^{1/2}.
\end{align}
This proves the desired result.
\end{proof}

\section{Short range quantum correlations}\label{appendix:eqv}

\begin{proposition}
Consider the special case of Definition~2.3 with two parties on Alice's side,
$A_0,A_1$, and one party on Bob's side, $B_0$. Write
\[
A:=A_0A_1,\qquad B:=B_0,
\]
and let
\[
q(a,b|x,y):=\omega^{(0,0)}\!\left(M^0_{a|x}N^0_{b|y}\right)
\]
denote the behaviour of the key-generating pair $A_0B_0$ in the round
$s=(0,0)$ (where we view $M^0_{a|x}$ as an observable on the full Alice
system $A$ by tensoring with the identity on $A_1$).

Then the following are equivalent:

\begin{enumerate}
\item The behaviour $q$ admits a realization with $\omega^{(0,0)}$ separable
across $A$--$B$.

\item The behaviour $q$ admits a representation
\[
q(a,b|x,y)=\sum_{\lambda} p(b|y,\lambda)\,
\rho\!\left(M^0_{a|x}\otimes G_\lambda\right)
\]
for some state $\rho$ on $A\otimes B$ and some POVM $\{G_\lambda\}_\lambda$
on $B$.
Equivalently, Bob's measurements in the key round are jointly measurable.
\end{enumerate}

Hence, in this special case, the locality condition defining
$\mathrm{LQC}_{A_0B_0}$ coincides with the short-range condition of
Lobo--Pauwels--Pironio \cite{Lobo2024}.
\end{proposition}

\begin{proof}
$(1)\Rightarrow(2)$:
Assume that $\omega^{(0,0)}$ is separable across $A$--$B$. In a concrete
representation we may write
\[
\omega^{(0,0)}=\sum_i p_i\,\alpha_i\otimes \beta_i
\]
as a convex combination of product states on $A$ and $B$. Therefore
\[
q(a,b|x,y)
 =\sum_i p_i\,\alpha_i(M^0_{a|x})\,\beta_i(N^0_{b|y}).
\]
Introduce a classical register with orthogonal projections $\{e_i\}_i$ and
define
\[
G_i:=e_i,\qquad p(b|y,i):=\beta_i(N^0_{b|y}).
\]
Then $\{G_i\}_i$ is a parent POVM and
\[
q(a,b|x,y)
 =\sum_i p(b|y,i)\,\tilde\omega(M^0_{a|x}e_i),
\]
where
\[
\tilde\omega:=\sum_i p_i\,\alpha_i\otimes \delta_i,
\]
with $\delta_i(e_j)=\delta_{i,j}$, is the obvious cq-state on $A\otimes \ell^\infty(i)$. Thus, Bob's measurements
are jointly measurable.

$(2)\Rightarrow(1)$:
Assume now that
\[
q(a,b|x,y)=\sum_{\lambda} p(b|y,\lambda)\,
\rho(M^0_{a|x}\otimes G_\lambda)
\]
for a POVM $\{G_\lambda\}_\lambda$ on $B$. Set
\[
p_\lambda:=\rho(\mathds{1}\otimes G_\lambda),
\qquad
\alpha_\lambda(X):=
\frac{\rho(X\otimes G_\lambda)}{p_\lambda}
\quad (p_\lambda>0).
\]
On the classical algebra $\ell^\infty(\lambda)$ with minimal projections
$\{e_\lambda\}_\lambda$ define
\[
\tilde N_{b|y}:=\sum_\lambda p(b|y,\lambda)e_\lambda
\]
and the separable state
\[
\tilde\omega:=\sum_\lambda p_\lambda\,\alpha_\lambda\otimes \delta_\lambda .
\]
Then
\[
\tilde\omega(M^0_{a|x}\tilde N_{b|y})
 =\sum_\lambda p(b|y,\lambda)\,\rho(M^0_{a|x}\otimes G_\lambda)
 =q(a,b|x,y).
\]
Hence the same behaviour is generated by a separable state across $A$--$B$.
\end{proof}
\end{document}